%% file: main.tex
\documentclass[letterpaper]{llncs}
\usepackage[margin=1in]{geometry}

\usepackage{times} 
\usepackage{helvet}  
\usepackage{courier}  
\usepackage{url}
\usepackage{graphicx}
\usepackage{float}
\usepackage{mathrsfs}
\usepackage{multirow}
\usepackage{amsmath,bm}
\usepackage{nicefrac}

\usepackage{bibnamed}
\usepackage{soul}
\usepackage{etoolbox}
\patchcmd{\thebibliography}{\section*{\refname}}{}{}{}

\usepackage{graphicx,color}
\usepackage{colortbl}
\usepackage{amssymb}
\usepackage{lipsum} 
\usepackage{enumitem}

\usepackage{enumitem}

\usepackage[linesnumbered,ruled]{algorithm2e}

\input{macros.tex}

\newtheorem{myclaim}{Claim}

\title{How Private Are Commonly-Used Voting Rules?}
\author{
 Ao Liu$^{1}$ \and
Yun Lu$^{2}$ \and
Lirong Xia$^{1}$ \and
Vassilis Zikas$^{3}$}
\institute{
	Rensselaer Polytechnic Institute, {\tt liua6@rpi.edu}, {\tt xial@cs.rpi.edu} \and
	University of Edinburgh, {\tt y.lu-59@sms.ed.ac.uk} \and
	Purdue University, {\tt vzikas@purdue.edu}
}

\begin{document}
\maketitle

\definecolor{red}{rgb}{0,0,0}
\definecolor{blue}{rgb}{0,0,0}

\begin{abstract}
	\input{00-abstract.tex}

\end{abstract}

\section{INTRODUCTION}\label{section:introduction}
\input{01-intro.tex}

\section{PRELIMINARIES}\label{section:preliminaries}
\input{02-preliminaries.tex}

\section{DISTRIBUTIONAL DIFFERENTIAL PRIVACY FOR VOTING}
\input{03-DDP-for-voting.tex}

\section{EXACT PRIVACY OF VOTING RULES: TWO-CANDIDATE CASE}\label{section:exactDDP}
\input{04a-exactDDP.tex}

\input{04b-plurality-score-vs-winner.tex}

\section{EXACT PRIVACY OF VOTING RULES: GENERAL CASE}\label{sec:GSRmore}
\input{05-privacy-GSR.tex}

\subsubsection*{Acknowledgements}
	We thank all anonymous reviewers for helpful comments and suggestions. 
	Vassilis Zikas: Work done while the author was at RPI and at the University of Edinburgh. Research supported in part by Sunday Group, Inc.,  and part by the Office of the Director of National Intelligence (ODNI), Intelligence Advanced Research Projects Activity (IARPA), via 2019-1902070008. The views and conclusions contained herein are those of the authors and should not be interpreted as necessarily representing the official policies, either expressed or implied, of ODNI, IARPA, or the U.S. Government. The U.S. Government is authorized to reproduce and distribute reprints for governmental purposes notwithstanding any copyright annotation therein. Lirong Xia: acknowledges NSF \#1453542 and \#1716333, and ONR \#N00014-17-1-2621 for support. Ao Liu: acknowledges an IBM AIHN scholarship for support.

\section*{References}
{\small
\bibliographystyle{named}
\input{main.bbl}}
\section*{Appendix References}
{\small
	[Mossel {\em et al.}, 2013] Elchanan Mossel, Ariel D. Procaccia, and Miklos Z. Racz. A Smooth Transition From Powerlessness to Absolute Power. {\em Journal of Artificial Intelligence Research}, 48(1):923–951, 2013
}
\onecolumn
\newpage
\appendix
\input{app-related_literature.tex}
\input{app-appendix.tex}

\end{document}

%% file: macros.tex
\newcommand\tab[1][1cm]{\hspace*{#1}}
\allowdisplaybreaks

\newcommand{\setvar}[1]{\mathcal{#1}}
\newcommand{\funvar}[1]{\mathbf{#1}}

\newcommand{\citeay}[1]{\citeauthor{#1}~(\citeyear{#1})}

\newcommand{\trail}[4]{{\mathsf T}_{#1,#2,#3,#4}} 
\newcommand{\ent}[1]{\text{Enter}(#1)} 
\newcommand{\ext}[1]{\text{Exit}(#1)} 

\newcommand{\constant}[1]{\Theta(#1)}
\newcommand{\const}{\Theta}
\newcommand{\Dist}{\text{Dist}}
\newcommand{\gdist}{\text{Dist}^*}

\newcommand{\bins}{l} 
\newcommand{\mechset}{\setvar{M}}

\newcommand{\ml}{\mathcal L}
\newcommand{\mc}{\mathcal C}
\newcommand{\mr}{\mathcal R}
\newcommand{\mt}{\mathsf T}
\newcommand{\sign}{\text{Sign}}
\newcommand{\piv}{\text{Piv}}
\newcommand{\ed}{\text{End}}
\newcommand{\pdf}{\text{PDF}}

\newcommand{\Sim}{\mathbf{Sim}}
\newcommand{\mech}{\mathbf{M}}
\newcommand{\Hist}{\funvar{Hist}}
\newcommand{\Supp}[1]{\mathbf{Supp}(#1)}

\newcommand{\setS}{\setvar{S}}
\newcommand{\universe}{\setvar{U}}

\newcommand{\rangeM}{\setvar{R}}

\newcommand{\altDDP}{DDP}
\newcommand{\exactDDP}{Exact Distributional Differential Privacy}
\newcommand{\eDDP}{eDDP}

\newcommand{\iidist}[2]{X_{#1, #2}}
\newcommand{\Natural}{\mathbb{N}}
\newcommand{\Real}{\mathbb{R}}
\newcommand{\LinOrd}[1]{\mathcal{L}(#1)}
\newcommand{\f}{\funvar{f}}
\newcommand{\util}{\funvar{u}}

%% file: 00-abstract.tex
{\em Differential privacy} has been widely applied to provide privacy guarantees by adding random noise to the function output. However, it inevitably fails in many high-stakes voting scenarios, where  voting rules are required to be deterministic. 
In this work, we present the first framework for answering the question:
{\em``How private are commonly-used voting rules?"} 
Our answers are two-fold. First, we show that deterministic voting rules provide sufficient privacy in the sense of 
{\em distributional differential privacy (DDP)}. We show that assuming the adversarial observer has uncertainty about individual votes, even publishing the histogram of votes 
achieves good DDP.
Second, we introduce the notion of {\em exact} privacy to compare the privacy preserved in various commonly-studied voting rules, and obtain dichotomy theorems of exact DDP within a large subset of voting rules called {\em generalized scoring rules}. 

%% file: 01-intro.tex

Differential privacy (DP) has gained much public attention recently, partly due to  its use in the United States 2020 Census. Improving upon ad-hoc privacy techniques that were broken in the previous census~\cite{garfinkel2018understanding}, formal privacy definition like DP are much more suitable for controlling the leakage of sensitive data.

Yet,  sensitive data is still published today without necessarily understanding the privacy leakage it incurs. In particular, voting data has been surprisingly accessible. In the US, histograms of votes are revealed per county, and voting and registration tables are released~\cite{uscensusbureau_2019}, which include fields like sex, race, age, location, and marital status. 
This abundance of information has enabled politicians to buy voter profiles from data mining companies to manipulate public opinion~\cite{verini_2007,bradshaw2018challenging}. 

Unfortunately, it is not easy to achieve (differential) privacy for voting. 
It is insufficient to   protect voter registration tables with proven privacy techniques; releasing the election outcome can also be a cause of information leakage. To see how an individual's vote can be inferred by observing the winner of the election, we consider the following example. 
Suppose Alice cast a vote in an election, and then the winner is announced. Further suppose that an adversary can accurately estimate other votes from questionnaires or by machine learning from the other voters' social media, and it turns out these other votes ended up with a tie among the candidates.  In this case, the adversary can distinguish Alice's vote even if he knows nothing about Alice, 
since Alice must have voted for the winner as the tie-breaker.

The strict definition of differential privacy means the mere {\em possibility} of the above scenario is a privacy violation. Moreover, 
ties do occur quite often in real life elections. For example, 9.2\% of STV elections on Preflib election data~\cite{Mattei13:Preflib} are tied~\cite{Wang2019:Practical}. 
Even if we consider another formal privacy definition that accepts the uncertainly stemming from machine learning methods or low likelihood of ties as helpful in disguising votes, it is unclear how to quantitatively measure the effect of such uncertainty, and how (or whether) privacy differs for different voting rules.

Motivated by the privacy concern in voting,  we focus on the following key question in this paper. 

\vspace{2mm}
\hfill{\em How private are commonly-used voting rules?\hfill }
\vspace{2mm}

The importance of answering this question is both practical and theoretical. On the practical side, minimizing the amount of information leakage from voting rules helps protect against censorship, coercion, and vote buying.  
On the theoretical side, privacy provides a new angle to comparing voting rules and designing new ones. 

A first attempt would be to employ \emph{differential privacy (DP)}~\cite{Dwork2006}, measure of privacy widely-accepted and widely-applied in the cryptographic community. Mathematically, a voting rule $\mech$ for $n\in\mathbb N$ voters is a mapping $\mech:\universe^n\rightarrow \rangeM$, where $\universe$ is the set of all possible votes; $\rangeM$ is the set of all possible outcomes of voting, e.g.~winners or histograms of votes. 
$\mech$ is $(\epsilon,\delta)-${\emph{differentially private}} if for any pair of {\em preference profiles} $\vec X\in \universe^n$ and $\vec X'\in \universe^n$ that only differ on one vote, and any subset of outcomes $\setS\subseteq \rangeM$, the following inequality holds:
{\footnotesize
\begin{equation}\label{equ:dp1}
    \Pr\left[\mech(\vec X) \in \setS\right] \leq e^\epsilon \Pr\left[\mech(\vec X') \in \setS\right]+\delta.
\end{equation}}
Smaller $\epsilon$, $\delta$ are desirable as it means the outcome of $\mech$ is not affected much by one vote, and thus reveals little about an individual voter. Note in general $\mech$ must be randomized to satisfy Inequality~(\ref{equ:dp1}); indeed~\cite{Shang2014,Lee2015,Hay2017} achieved DP via randomized voting. 

Yet most, if not all, voting rules used in high-stakes political elections are deterministic, since randomized voting rules suffer from 
difficulties in verifying implementation correctness, e.g. the controversy in the 2016 Democratic primary election in Iowa~\cite{iowa2016}. 
Unfortunately, the randomness in Inequality (\ref{equ:dp1}) comes from the voting rule itself, so deterministic rules cannot achieve DP except with the trivial parameter of $\delta \geq 1$, which always holds (see Example~\ref{ex:dpfails} for more details).

\subsection{OUR CONTRIBUTIONS}\label{sec:ourcontribution}
To overcome the critical limitation of DP in high-stakes voting scenarios, we study the privacy of deterministic voting rules using 
{\em distributional differential privacy (DDP)}~\cite{Bassily2013}, a well-accepted notion of privacy that works for deterministic functions. DDP measures the amount of individual information leakage, assuming the adversary only has uncertain information about voter preferences, for example when using a machine learning algorithm. 
Our result on the DDP of commonly-used voting rules carries the following encouraging message:

\vspace{2mm}
\noindent {\bf Main Message 1: Many commonly-used voting rules achieve good DDP in natural settings. }
\vspace{2mm}

More precisely, we focus on a natural DDP setting where the adversary's information is represented by a set of i.i.d.~distribution's over preference profiles, denoted by $\Delta\subseteq \Pi(\universe)$, where $\Pi(\universe)$ is the set of all probability distributions over $\universe$ {\em with full support}. A voting rule $\mech$'s DDP is now measured by three parameters $(\epsilon,\delta,\Delta)$. 
A deterministic function is DDP (Definition~\ref{def:our-DDP}) if it satisfies an inequality similar to Inequality (\ref{equ:dp1}), but now the randomness is replenished by the adversary's uncertainty about the profile $\vec X$, represented by $\Delta$. Like DP, smaller $\epsilon$ and $\delta$ in DDP are more desirable.

With DDP, we can quantitatively measure the privacy of the histogram rule $\Hist$, which outputs the frequency of each type of vote in the preference profile, in the following Theorem~\ref{thm:epsdelta-histogram-c}. As an immediate consequence, many common voting rules also achieve good privacy.
\noindent\begin{paragraph}
{\bf Theorem~\ref{thm:epsdelta-histogram-c} (DDP for $\Hist$).} {\em Given any $\universe=\{x_1,\ldots,x_{\bins}\}$ and $\Delta\subseteq \Pi(\universe)$ with $|\Delta|<\infty$, let $p_{\min} = \min_{\pi\in\Delta, i\le \bins} (\pi(x_i))$. For any $n\in\Natural$ and any $\epsilon \geq 2 \ln\left(1+\frac{1}{p_{\min}n}\right)$, $\Hist$ for $n$ voters is $(\epsilon,\delta,\Delta)$-DDP where $\delta =\exp(-\Omega(np_{\min}[\min(2\ln(2), \epsilon)]^{2}))$.}

\end{paragraph}

{Theorem~\ref{thm:epsdelta-histogram-c}} states that  $\Hist$ is private with good parameters, as even a small $\epsilon$ results in $\delta$ that is considered {\em negligible} in cryptography literature. The winner of many commonly-used voting rules  depends only on the outcome of $\Hist$, and thus contain (often strictly) less information than $\Hist$. Thus, they achieve {\em at least as good} privacy w.r.t.~DDP as simply outputting the histogram.

Next, we highlight that DDP (as well as DP and its variants) parameters only describe loose bounds on privacy---by definition, if a voting rule satisfies $(\epsilon, \delta, \Delta)$-DDP, it also satisfies $(\epsilon + 0.1, \delta + 0.1, \Delta)$-DDP. To compare the privacy-preserving capability of voting rules, we introduce 
the notion of {\em exact distributional differential privacy (eDDP)}, whose parameters describe tight bounds on $\epsilon$ and $\delta$.  
We focus on the $\epsilon = 0$ case as a first step to compare various voting rules with their eDDP in the $\delta$ parameter.  
Our results on the eDDP of commonly-used voting rules carry the following message: 

\vspace{2mm}
\noindent {\bf Main Message 2: For many combinations of commonly-used voting rules and $\Delta$, the $(0,\delta,\Delta)$-eDDP exhibits a dichotomy between $\delta=\Theta(\sqrt{1/n})$ and $\delta=\exp(-\Omega(n))$. }
\vspace{2mm}

More precisely, we prove the following dichotomy theorem for two candidates $\{a,b\}$ and {\em $\alpha$-biased majority rules} with $\alpha\in (0,1)$, which chooses  $a$ as the winner iff  at least $\alpha n$ out of $n$ votes prefer $a$.

\noindent\begin{paragraph}
{\bf Theorem~\ref{thm:nonunif-winner-2} (Dichotomy in Exact DDP for $\alpha$-Majority Rules over Two Candidates, Informal)}
{\em Fix two candidates $\{a, b\}$ and $\Delta\subseteq \Pi(\{a, b\})$ with $|\Delta|<\infty$. For any $\alpha\in(0,1)$, the $\alpha$-biased majority rule is $(0,\delta,\Delta)$-eDDP for all $n$, where $\delta$ is either $\Theta(\sqrt{1/n})$, when $\Delta$ contains a distribution $\pi$ with $\pi(a) = \alpha$, or exponentially small otherwise. 
}

\end{paragraph}

For more than two candidates, we prove the following dichotomy theorem for a large family of voting rules and $\Delta\subseteq \Pi(\universe)$. 

\begin{paragraph}{\bf Theorem~\ref{thm:gsr} (Dichotomy in Exact DDP of A Large Class of Voting Rules and $\Delta$, Informal)} 
{\em For any fixed number of candidates, and any voting rule in a large family, the $\left(0, \delta,\Delta\right)$-eDDP is $\delta = \Theta(\sqrt{1/n})$, when $\Delta$ contains the uniform distribution, or $\delta =\exp(-\Omega(n))$, when $\Delta$ is finite and does not contain any unstable distributions.
}
\end{paragraph}

Intuitively, a distribution $\pi$ is {\em unstable} under a voting rule $\mech$ if adding small perturbations 
can cause a different candidate to win (Definition~\ref{dfn:unstable}). Instead of conducting case-by-case studies of eDDP for commonly-used voting rules, we prove Theorem~\ref{thm:gsr} for a large family of voting rules called {\em generalized scoring rules}~\cite{Xia2008} that further satisfy {\em monotonicity, local stability}, and {\em canceling-out}. We show that many commonly-used voting rules satisfy these conditions (Section~\ref{sec:GSRmore}). We also compute and compare the concrete $\delta$ values for small elections 
 (Table~\ref{GSR-concrete-delta}, Section~\ref{section:empirical-results-main} and Appendix~\ref{section:empirical-results}).

\begin{table*}[htp]
    \centering
\resizebox{\textwidth}{!}{
	\begin{tabular}{ |c|c|c|c|c|c|} 
		\hline
			$\mech$	& Borda & STV & Maximin & Plurality & 2-approval \\
		\hline		
			 $\delta(n)$	& $
			 \dfrac{1}{\sqrt{1.347n + 0.5263}}$  
				& 
				$
				\dfrac{1}{\sqrt{1.495n + 0.02669}}$ 
				&
				 $
				 \dfrac{1}{\sqrt{1.553n + 4.433}}$
				&
				$
				\dfrac{1}{\sqrt{1.717n - 0.09225}}$ 
				&
				$
				\dfrac{1}{\sqrt{1.786n + 0.3536}}$ \\ 
		\hline
	\end{tabular}
}
    \vspace{1mm}
	\caption{{\small $\delta$ values in $(0, \delta, \Delta)$-eDDP for some commonly-used voting rules under the i.i.d.~uniform distribution, $m=3$ and $n\leq50$. From left to right, we rank rules from least to most private.}}
	\label{GSR-concrete-delta}
\end{table*}

\subsection{RELATED WORK}

Differential privacy~\cite{Dwork2006} has been used to add privacy to rank aggregation: \citeay{Shang2014} applied Gaussian noise to the histogram of linear orders, while~\citeay{Hay2017} used Laplace and Exponential mechanisms applied to specific voting rules. \citeay{Lee2015} also developed a method of random selection of votes to achieve differential privacy. 
One interesting aspect of adding noise to the output that was observed in~\cite{Birrell2011,Lee2015} is that it enables an approximate strategy-proofness; the idea here is that the added noise dilutes the effect of any individual deviation, thereby making strategies which would slightly perturb the outcome irrelevant. We remark that if one wishes to achieve DP for a large number of voting rules, well-known DP mechanisms (like adding Laplace noise~\cite{Dwork2006a}) can be applied to rules in GSR in a straightforward way, by adding noise to each component of the score vector and outputting the winner based on the noised score vector. Our work is different because we focus on exact privacy of deterministic voting rules.

 In our work, we compare deterministic functions by their exact privacy. In differential privacy literature where functions must be randomized, their accuracy, or utility, is used to compare them. A number of works have defined utility as a metric which describes the comparative desirability of $\epsilon$-DP mechanisms. In~\cite{McSherry2007}, utility is an arbitrary user-defined function, used in the exponential mechanism. The works of~\cite{Blum2008,Hardt2010,Bassily2015} define utility in terms of error, where the closer (by some metric) the output of the function, which uses this mechanism to apply noise, is from the desired (deterministic) query's, the higher the utility; the definition of~\cite{Ghosh2009} in addition allows the user to define as a parameter, the prior distribution on the query output.
 In contrast, our results imply that in the context of distributional differential privacy, voting rules achieve a well-accepted notion of privacy while preserving perfect accuracy, or utility.
 
\subsection{DISCUSSIONS}
While DP 
has been widely applied to measure privacy and has been applied to voting, as we discussed in the Introduction, it fails for deterministic functions such as voting rules in high-stakes  elections. This critical limitation motivates our study. To the best of our knowledge, we are the first to illustrate how to measure privacy in high-stakes voting using (e)DDP in a natural setting. We will see that the problem, though challenging, can be solved by our novel {\em trails} technique. Below we explicitly discuss our conceptual and technical contributions and closely related works. More comprehensive discussions of related work can be found in Appendix~\ref{appendix:related_literature}.

\begin{paragraph}{\bf Conceptual contributions.} Our first conceptual contribution is the application of DDP to deterministic voting rules. As discussed earlier, while previous works add random noise to achieve DP, to the best of our knowledge, no previous studies were done for deterministic voting rules. We note that the {\em truncated} histogram result of~\cite{Bassily2013} does not suffice, since in general, votes are not removed in an election.
Moreover, we prove our results in a simpler definition than DDP; the equivalence of this definition and DDP is proven in Appendix~\ref{sec:equivalence}. Our second conceptual contribution is the introduction of {\em exact DDP}, addressing the issue that parameters of DDP (and other relaxations of DP~\cite{Bassily2013,Groce2014,Kasiviswanathan2008,Hall2012,Duan2009,Bhaskar2011}) describe only upper bounds on privacy. We are not aware of other works that explicitly propose to characterize tight bounds on the privacy parameters $\epsilon$ and $\delta$.
\end{paragraph}

\begin{paragraph}{\bf Technical contributions.} Our first theorem (Theorem~\ref{thm:epsdelta-histogram-c}) is quite positive, showing the privacy of outputting histograms. Theorem~\ref{thm:nonunif-winner-2} and~\ref{thm:gsr} characterize eDDP in terms of $\delta$ values by fixing $\epsilon=0$. We do so for the two reasons: (1) it is the common convention to compute $\delta$ based on a fixed $\epsilon$ for DP or DDP; (2) $\epsilon=0$ is the most informative choice, since Theorem~\ref{thm:epsdelta-histogram-c} shows that even for small non-zero $\epsilon$, any difference we can observe in the $\delta$ of two voting rules is exponentially small---considered negligible in cryptography literature. 
While our theorems appear similar and related to the dichotomy theorems on the probability of ties in voting~\cite{Xia2008,Xia2015}, the definition and mathematical analysis are quite different, and previous techniques do not work for all cases; see more discussions in the proof sketch for Theorem~\ref{thm:gsr}. 
 To address the challenge, we developed the {\em trails} technique, which significantly simplifies  calculations. 
\end{paragraph}

\begin{paragraph}{\bf Generality of our setting.}
As the first work towards answering our key question, we assume the adversary's beliefs are modeled by a set of  i.i.d.~distributions over the votes. A special case is the i.i.d.~uniform distribution, which is known as the {\em impartial culture} assumption in social choice~\cite{Georges-Theodule1952:Les-theories}. Extending to general $(\epsilon$, $\delta)$, and non-i.i.d.~distributions is an important and challenging future direction.
Lastly, though our definitions and results are presented in the context of voting for the sake of presentation, they can easily be extended to general applications.
\end{paragraph}

%% file: 02-preliminaries.tex

Let $\mathcal C=\{c_1,\ldots, c_m\}$ be a set of $m\ge 2$ candidates, and $\mathcal{L}({\mathcal C})$ denote the set of all {\em linear orders} over $\mathcal C$: that is, the set of all antisymmetric, transitive, and total binary relations.  Let $\universe$ denote the set of all possible votes. Given $n\in \Natural$, we let $\vec X=(X_1,\ldots, X_n)\in \universe^n$ denote a collection of $n$ votes called a {\em preference profile}. Let $\rangeM$ denote the set of outcomes of voting. A (deterministic) voting rule for $n$ voters is a mapping $\mech: \universe^{n}\rightarrow \rangeM$.

For example, in the {\em plurality} rule, $\universe = \rangeM=\mathcal C$; each voter votes for one favorite candidate, and the winner is the candidate with the most votes. In the {\em Borda} rule, $\universe = {\mathcal L}( \mathcal C)$  and $\rangeM = \mathcal C$; each voter cast a linear order $X$ over $\mathcal C$, denoted by $c_{i_1}\succ c_{i_2}\succ \cdots \succ c_{i_m}$, where $a\succ b$ means that $a$ is preferred over $b$; each candidate $c$ gets $m-i$ points in each vote, where $i$ is the rank of $c$ in the vote; the winner is the candidate with the highest total points. A tie-breaking mechanism is used when there are ties in plurality and Borda.

\begin{definition}[The histogram rule] Let $\universe = \{x_{1}, \cdots, x_{\bins}\}$. For any $n\in \mathbb \Natural$, the {\em histogram function}, denoted by $\Hist:\universe^{n}\rightarrow \Natural^{\bins}$, takes as input a preference profile $\vec X = (X_1,\ldots,X_n)\in \universe^{n}$ and outputs a $\bins$-dimensional integer vector whose $i$th component is $|\{j\colon X_j = x_i, j\in \{1,\cdots, n\}\}|$. 
\end{definition}
For example, when applied to the setting of the plurality rule, $\bins=m$ and  $\Hist$ outputs the number of votes each candidate receives. When applied to the setting of the Borda rule, $\bins=m!$ and  $\Hist$ outputs the number of occurrences of each linear order.

%% file: 03-DDP-for-voting.tex
As we discussed, DP is not a suitable notion to analyze nontrivial deterministic voting rules as shown in the following example, which motivates our use of distributional differential privacy (DDP)~\cite{Bassily2013}. 

\begin{example}[DP fails for deterministic voting rules]\label{ex:dpfails} Consider the plurality rule for two candidates $\{a,b\}$ and three voters ($n=3$). We have $\universe =\rangeM =\{a,b\}$. In Inequality~(\ref{equ:dp1}), let $\vec X=(a,a,b)$, $\vec X' = (b,a,b)$, and $\setS=\{a\}$. Then, (\ref{equ:dp1}) becomes $1\le e^\epsilon\times 0+\delta$, which means that $\delta \ge 1$.
\end{example}

 At a high level, the DDP of a (deterministic or randomized) function is characterized by three parameters $(\epsilon, \delta, \Delta)$, where $\epsilon$ and $\delta$ are privacy parameters similar to DP, and $\Delta$ is a set describing the adversary's knowledge about the preference profile.  We consider  adversaries that can be modeled as $\Delta \subseteq \Pi(\universe)$, which encodes each of the adversary's possible uncertainties as a distribution where each vote is i.i.d.. 

\begin{example}[Adversary's information $\Delta$]\label{ex:Delta} Suppose $\universe =\rangeM={\mathcal C} =\{a,b\}$, and 
the $n$ votes could be i.i.d.~generated from either $\pi_{0.2}$ or $\pi_{0.7}$. 
Here, for any $\gamma\in [0,1]$, $\pi_\gamma (a) = \gamma$.  Then, the adversary's information is  represented by $\Delta = \{\pi_{0.2}, \pi_{0.7}\}$. Say we  prove that some voting rule is $(\epsilon = 0.5, \delta = 0.1, \Delta)$-DDP for the above $\Delta$. Intuitively, this means that the voting rule has privacy  $\epsilon = 0.5, \delta = 0.1$, given the adversary's knowledge can be modeled by any distribution in $\Delta$. We remark that this privacy holds {\em without} the need to add noise to the outcome of the election, contrasting with DP. 
\end{example}

To simplify presentation, below we will introduce the definition of DDP studied in this paper. In our setting of this paper, our simpler definition is equivalent to the original DDP. More details can be found in  Appendix~\ref{sec:equivalence}.

\begin{definition}[DDP studied in this paper]\label{def:our-DDP} For any $\Delta\subseteq \Pi(\universe)$, $\epsilon>0$, and $\delta>0$, a voting rule $\mech:\universe^{n}\rightarrow \rangeM$ is $(\epsilon, \delta, \Delta)$-DDP if for every $\pi\in \Delta$,   $i\le n$,  $x, x'\in \universe$, and  $\setS \subseteq \rangeM$, the following inequality holds.
{\footnotesize
\begin{equation}\label{dfn:ddp}
\begin{split}
&\Pr\nolimits_{\vec X\sim \pi}(\mech(\vec X)\in \setS | X_i = x) \\ \leq\;& e^{\epsilon} \Pr\nolimits_{\vec X\sim \pi}(\mech(\vec X)\in \setS |  X_i = x') + \delta,
\end{split}
\end{equation}}
where $\vec X =(X_1,\ldots, X_n) $ is a preference profile where each vote is i.i.d.~generated from $\pi$.
\end{definition}
For deterministic $\mech$, the randomness in Inequality (\ref{dfn:ddp}) comes from the adversary's incomplete information, captured by $\Delta$. We show that $\Hist$ satisfies good DDP.

\begin{theorem}[DDP of $\Hist$, proof in Appendix~\ref{appendix:thm1}]\label{thm:epsdelta-histogram-c} Given any $\universe=\{x_1,\ldots,x_{\bins}\}$ and $\Delta\subseteq \Pi(\universe)$  with $|\Delta|<\infty$, let $p_{\min} = \min_{\pi\in\Delta, i\le \bins} (\pi(x_i))$. For any $n\in\Natural$ and any $\epsilon \geq 2 \ln(1+\frac{1}{p_{\min}n})$, $\Hist$ for $n$ voters is $(\epsilon,\delta,\Delta)$-DDP where $\delta =\exp(-\Omega(np_{\min}[\min(2\ln(2), \epsilon)]^{2}))$.
\end{theorem}

As corollary, these privacy parameters of $\Hist$ automatically apply to all functions that only depend on the output of $\Hist$, i.e.~most voting rules, or outputting the histogram in addition to the winner as in US presidential elections. This follows immediately from a property of DDP called {\em immunity to post processing} (see Lemma~\ref{lem:foM} in Appendix~\ref{appendix:thm1}). We note the result is similar to that of~\cite{Bassily2013}, but they assume  lower-frequency items in the histogram are truncated (which is not the case in general when  election results are posted) and describe a less precise $\delta$.

%% file: 04a-exactDDP.tex

In this section, we first present the definition of {\em exact distributional differential privacy} (exact DDP or eDDP), then characterize $(0,\delta,\Delta)$-eDDP for two candidates under any $\alpha$-biased majority rule. The proof of this theorem will serve as a toy application of our {\em trails technique}, useful for proving our main result Theorem~\ref{thm:gsr}. 

Intuitively, a function has \emph{exact privacy} with parameters $\epsilon$ and $\delta$ if the function cannot satisfy the privacy definition with strictly better parameters. We remark that this definition can easily be altered to define \emph{$(\epsilon, \delta)$-exact differential privacy (eDP)} by omitting $\Delta$. 

\begin{definition}[\exactDDP~(\eDDP)]\label{def:exactDDP}
	A voting rule $\mech$ is \emph{$(\epsilon, \delta, \Delta)$-\exactDDP~(\eDDP)} if it is $(\epsilon, \delta, \Delta)$-DDP~ and there does not exist $(\epsilon'\leq \epsilon, \delta' < \delta)$ nor $(\epsilon' < \epsilon, \delta' \leq \delta)$ such that $\mech$ is $(\epsilon', \delta', \Delta)$-DDP.
\end{definition}

The {\em $\alpha$-biased majority rule}, denoted by $\mech_\alpha$, over two candidates $(a, b)$ outputs $a$ as the winner if at least $\alpha$ fraction of votes prefer $a$ over $b$. An example of this type of voting rule is {\em supermajority}, used in government decisions around the world.

\begin{theorem}[Exact DDP for Majority Rules, full proof in Appendix~\ref{sec:fullproof_thm2}]\label{thm:nonunif-winner-2}
Fix two candidates $\{a, b\}$ and $\Delta\subseteq \Pi(\{a, b\})$ with $|\Delta|<\infty$. 
For any $\alpha\in(0,1)$, the $\alpha$-biased majority rule is $(0,\delta,\Delta)$-eDDP for all $n$, where 
{\small $$\delta = \max_{p = \pi(a)\colon \pi\in\Delta}\Theta\left(\sqrt{\frac{1}{n}}\left[\left(\frac{p}{\alpha}\right)^{\alpha}\left(\frac{1-p}{1-\alpha}\right)^{1-\alpha} \right]^n\right).$$ }
In particular, $\delta = \Theta\left(\sqrt{1/n}\right)$ if there exists $\pi\in\Delta$with $\pi(a)=\alpha$; otherwise $\delta = \exp(-\Omega(n))$.
\end{theorem}

In the following subsections, we will present our {\em trails} technique for analyzing DDP in voting, followed by a proof sketch of Theorem~\ref{thm:nonunif-winner-2} using the trails technique.

%% file: 04b-plurality-score-vs-winner.tex

\subsection{OUR TOOL TO ANALYZE PRIVACY: TRAILS TECHNIQUE} 
Let us describe the trails technique using a simple, toy example: suppose there are two candidates $\{a, b\}$, and $n =5$ votes. Let $\mech$ be the majority rule where ties are broken in favor of $a$, i.e.~$\alpha=0.5$. 
We want to compute $(0,\delta,\Delta)$-eDDP of $\mech$ for any $\Delta\subseteq \Pi(\{a, b\})$.
In light of Definitions~\ref{def:our-DDP} and~\ref{def:exactDDP}, we have:
{\small
\begin{equation}\label{equ:delta}
\begin{split}
\delta =&\max_{\setS,x,x',i,\pi \in \Delta}\left[\Pr\nolimits_{\vec X\sim \pi}(\mech(\vec X)\in \setS | X_{i} = x)\right. \\
&\left.-\Pr\nolimits_{\vec X\sim \pi}(\mech(\vec X)\in \setS | X_{i} = x')\right].
\end{split}
\end{equation}}
 
Now, the majority rule is {\em anonymous}, that is, the identity of the voter is irrelevant and it chooses the winner only based on the histogram of votes. We can thus write $\mech = \funvar{f} \circ \Hist$, where $t=(t_a,t_b)$ and $\funvar{f}(t)$ outputs $a$ if $t_a \geq t_b$ and outputs $b$ otherwise. Then, Equation (\ref{equ:delta}) can be rewritten with probabilities over histograms, which is easier to compute (below, $\vec X\sim \pi$ is implicit).
 
{\footnotesize
\begin{equation}\label{eqn:delta-hist}
    \begin{split}
    \delta &=\max_{\setS,x,x',i,\pi \in \Delta}\left[\Pr(\funvar{f}(\Hist(\vec X))\in \setS | X_{i} = x) \right.\\
    &\;\;\;\;\;\;\;\;\;\;\;\;\;\;\;\;\;\;\;\;\;\;\;\;\left.-\Pr(\funvar{f}(\Hist(\vec X))\in \setS | X_{i} = x')\right]\\
    &= \max_{\setS,x,x',i,\pi \in \Delta}\left[\sum_{t \colon \funvar{f}(t) \in \setS} \Pr(\Hist(\vec X) = t| X_{i} = x) \right.\\
    & \;\;\;\;\;\;\;\;\;\;\;\;\;\;\;\;\;\;\;-\left. \sum_{t \colon \funvar{f}(t) \in \setS}\Pr(\Hist(\vec X) = t | X_{i} = x')\right].
    \end{split}
\end{equation}
}

For example, if $\setS = \{a\}$, then $\mt \equiv \{t \colon \funvar{f}(t) \in \setS\} = \{(5, 0), (4, 1), (3, 2)\}$ is an example of what we call a {\em trail}. Intuitively, a trail $\mt$ is a set of histograms {\em consecutive} in the sense that, starting from some $t$, we can list exactly the elements of $\mt$ by iteratively subtracting 1 from and adding 1 to two components of $t$, respectively.
 We see that $\mt$ can be listed in such a way, starting from entry $\ent{\mt} = (5, 0)$ and ending at exit $\ext{\mt} = (3, 2)$, by interatively subtracting from the first component and adding to the second component of $(5, 0)$ (we say the {\em direction} of $\mt$ is $(1, 2)$). See Figure~\ref{fig:trail-example-simple}.
 
\begin{figure}[ht]
   \centering
   \includegraphics[width=0.47\textwidth]{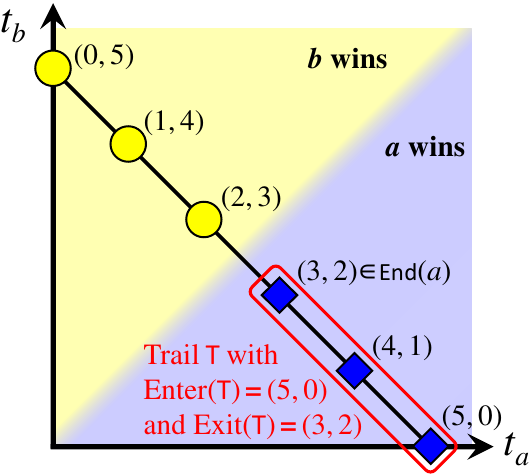}
   \caption{ \small{A  trail for two candidates. {A graph of number of votes for candidate $a$ ($=t_a$) versus votes for candidate $b$ ($=t_b$). Each point in the line is a histogram where the total number of votes is $n = 5$. The set $\{(5, 0), (4, 1), (3, 2)\}$ forms a trail. We denote by $\ed(a)$ (used in the proof of Theorem~\ref{thm:gsr}) the set of histograms which are exits of trails where $a$ is the winner. In this example $\ed(a) = \{(3,2)\}$.}} \label{fig:trail-example-simple}}
\end{figure}

We now give intuition for our key Lemma~\ref{lem:trail} presented below using this example. Suppose in Equation (\ref{eqn:delta-hist}) the maximizing $\setS$ is $\{a\}$ (so that $\{t \colon \funvar{f}(t) \in \setS\} = \mt$), $x = a$, and $x' = b$. Then, for any $i$, and any $\pi\in\Delta$:
{\small 
\begin{align*}
   \delta&=\sum_{t\in \{(5, 0), (4, 1), (3, 2)\} } \Pr(\Hist(\vec X) = t|  X_{i} = a) \\
   &\;\;\;\;\;\;- \sum_{t\in \{(5, 0), (4, 1), (3, 2)\}} \Pr(\Hist(\vec X) = t |  X_{i} = b).
\end{align*}}
The core of Lemma~\ref{lem:trail} is the observation that when votes are independent (e.g. when $\Delta\subseteq \Pi(\{a, b\})$), then for all $t = (t_a, t_b)$ such that $t_a > 0$, the following holds 
{\footnotesize
\begin{equation}\nonumber
\begin{split}
&\Pr(\Hist(\vec X) = (t_a, t_b)|X_i = a) \\
= &\Pr(\Hist(\vec X) = (t_a-1, t_b + 1) | X_i = b).
\end{split}
\end{equation}}
In light of this,
$\Pr(\Hist(\vec X) = (5, 0)| X_{i} = a)$ cancels out with $\Pr(\Hist(\vec X) = (4, 1)|  X_{i} = b)$, and $\Pr(\Hist(\vec X) = (4, 1)|  X_{i} = a)$ cancels out with $\Pr(\Hist(\vec X) = (3, 2)| X_{i} = b)$. This leaves
{\footnotesize
\begin{equation}\nonumber
\begin{split}
    \delta= &\Pr(\Hist(\vec X) = (3, 2) = \ext{\mt}|X_i = a)\\
    &- \Pr(\Hist(\vec X) = (5, 0) = \ent{\mt}|X_i = b).
\end{split}
\end{equation}}
We note that here $\Pr(\Hist(X) = \ent{\mt}|X_i = b) = 0$, but this does not hold generally for all trails for $m\ge 2$. This calculation can be extended to the more general Lemma~\ref{lem:trail} below.  
Before that, let us formally define trails. For any histogram $t = (t_{1}, \cdots, t_{\bins})\in \Natural^{\bins}$, any $z\in \mathbb{Z}$ and $j\le \bins$, we let $(t_{1}, \cdots, t_{\bins}) + zx_{j}$ denote the histogram $(t_{1}, \cdots, t_{j} + z,\cdots t_{\bins})$.

\begin{definition}[Trails]\label{def:trail}
Given a pair of indices $(j,k)$ where $j\ne k$, a histogram $t$, and a length $q$, we define the {\em trail} $\trail{t}{x_j}{x_k}{q}=\{t - zx_{j} + zx_{k}): 0\leq z \leq q\}$, where $(j,k)$ is called the direction of the trail, $t$ is then the {\em entry} of this trail, also denoted by $\ent{\trail{t}{x_j}{x_k}{q}}$, and $t - qx_{j} + qx_{k}$ is called the {\em exit} of the trail, denoted by $\ext{\trail{t}{x_j}{x_k}{q}}$.

Alternatively, a trail $\mt$ can be defined by just its entry and exit. 
\end{definition}

\begin{lemma} 
\label{lem:trail}
	Let $\mt$ be a trail with direction $(j,k)$, and let $\pi\in \Pi(\universe)$. For any $i$, $x_{j}, x_{k} \in \universe$, we have:
	{\footnotesize 
	\begin{equation}\nonumber
	\begin{split}
	&\Pr_{\vec X\sim \pi}(\Hist(\vec X) \in \mt\ |\ X_{i} = x_{j})\\
	&-	\Pr_{\vec X\sim \pi}(\Hist(\vec X) \in \mt\ |\ X_{i} = x_{k}) \\
	=\;&\Pr_{\vec X\sim \pi}(\Hist(\vec X) = \ext{\mt}\ |\ X_{i} = x_{j})\\
	&-  \Pr_{\vec X\sim \pi}(\Hist(\vec X) =\ent{\mt}\ |\ X_{i} = x_{k}).
	\end{split}
	\end{equation}}
\end{lemma}

\begin{proof} 
    Fix distribution $\pi$ over $n$ votes, where each vote is independently distributed. For $\vec X\sim\pi$, denote $X_{-i}$ as the random variable $\vec X$ but without the $i$th vote.
	The equality in the lemma comes from the simple observation that when votes are independently distributed, for any histogram $t\in \Natural^{\bins}$ and any $j\in [\bins]$
	{\footnotesize $$
	\Pr_{\vec X \sim \pi}(\Hist(\vec X) = t|X_{i} = x_{j}) = \Pr_{\vec X \sim \pi}(\Hist(X_{-i}) = t - x_{j})
	$$}
	(Below, $\vec X \sim \pi$ is implicit). Let $q$ be the length of the trail. For any $0\leq z<q$, let $t_{z} = \ent{\mt}-zx_{j} + zx_{k}$. Then,
{\footnotesize \begin{align*}
	&\Pr(\Hist(\vec X) = t_{z}|X_{i} = x_{j}) \\
	&= \Pr(\Hist(X_{-i}) =  t_{z} - x_{j})\\
	&=\Pr(\Hist(\vec X) =t_{z} - x_{j} + x_{k}|X_{i} = x_{k})\\
	&= \Pr(\Hist(\vec X) =t_{z+1}|X_{i} = x_{k}).
\end{align*}}
In other words,
{\footnotesize
\begin{align*}
&\Pr(\Hist(\vec X) \in \mt|X_{i} = x_{j}) \\
&-	\Pr(\Hist(\vec X) \in \mt|X_{i} = x_{k})\\
=& \Pr(\Hist(\vec X) = t_{q}|X_{i} = x_{j}) \\
&- \Pr(\Hist(\vec X) = t_{0} | X_{i} = x_{k}) \\
\tab& + \sum_{0\leq z<q} \Big(\Pr(\Hist(\vec X) = t_{z}|X_{i} = x_{j})\\
&\tab- \Pr(\Hist(\vec X) = t_{z+1}|X_{i} = x_{k})\Big)\\
=& \Pr(\Hist(\vec X) = t_{q}|X_{i} = x_{j}) \\
&\tab- \Pr(\Hist(\vec X) = t_{0}|X_{i} = x_{k})\\ \tag{Every term in the summation of differences cancels out.}\\
=& \Pr(\Hist(\vec X) = \ext{\mt}|X_{i} = x_{j}) \\
&\tab- \Pr(\Hist(\vec X) =\ent{\mt}|X_{i} = x_{k})
\end{align*}}
\end{proof}

\paragraph{Remark.} In this subsection's example, no matter the $\setS$, the set $\{t \colon \funvar{f}(t) \in \setS\}$ forms one single trail, but this does not hold 
in general. Instead, to prove our main theorem we will partition this set into multiple trails, and apply Lemma~\ref{lem:trail} to simplify probabilities over each trail.

\subsection{A SIMPLE APPLICATION OF TRAILS TECHNIQUE: PROOF OF THEOREM~\ref{thm:nonunif-winner-2}}

\begin{proof}{[Proof sketch for Theorem~\ref{thm:nonunif-winner-2}, see Appendix~\ref{sec:fullproof_thm2} for the full proof].} 
 
For any $\pi\in \Delta$, let $p = \pi(a)$. Let trails $\mt_a = \left\{t: t=(k, n-k) ,\, k \geq \alpha n\right\}$ and  $\mt_b= \left\{t: t=(k, n-k) ,\, k < \alpha n\right\}$. 
It follows that any histogram in $\mt_a$ results in $a$ being the winner, and any in $\mt_b$ results in $b$ as the winner. Also,  Equation (\ref{eqn:delta-hist}) implies we should {\em not} consider $\setS=\{a, b\}$ nor $\setS=\emptyset$ as otherwise $\delta = 0$ (the default lower bound on $\delta$).  Thus, we only consider $\setS = \{a\}$ (when  winner is $a$, corresponding to trail $\mt_a$) or  $\setS = \{b\}$ (trail $\mt_b$). Then Equation (\ref{eqn:delta-hist}) becomes (we disregard the value of $i$ since votes are i.i.d.):
{\footnotesize 
\begin{equation}\nonumber
\begin{split}
\delta =& \max_{j\in \{a, b\},x,x'} \left[\Pr\nolimits_{\vec X\sim \pi}(\Hist(\vec X)\in \mt_j | X_i = x) \right.\\
&\left.- \Pr\nolimits_{\vec X\sim \pi}(\Hist(\vec X)\in \mt_j | X_i = x')\right]\ \ \  \text{(Equation (\ref{eqn:delta-hist}))}\\
=& \max_{j\in\{a, b\},x,x'} \left[\Pr(\Hist(\vec X) = \ext{\mt_j} | X_i = x)\right.\\
&\left.- \Pr(\Hist(\vec X) = \ent{\mt_j} | X_i = x')\right].\ \ \  (\text{Lemma~\ref{lem:trail}})\\
\end{split}
\end{equation}}
We first discuss $\setS =\{a\}$ whose corresponding trail $\mt_a$ starts at $\ent{\mt_a} = (n,0)$ and exits at $\ext{\mt_a} = (\lceil \alpha n \rceil, \lfloor (1-\alpha)n \rfloor)$. Here, $x = a$ and $x' = b$ maximize $\delta$. Then,
{\footnotesize 
\begin{align*}&\Pr(\Hist(\vec X) = \ent{\mt_a} | X_i = b) \\=\; &\Pr(\Hist(\vec X) = (n,0) | X_i = b)
=0,
\end{align*}}
 
and
 
{\footnotesize \begin{align*}\nonumber
\Pr(\Hist(\vec X) = \ext{\mt_a} | X_i = a) \\=  \Theta\left(\sqrt{\frac{1}{n}}\left[\left(\frac{p}{\alpha}\right)^{\alpha}\left(\frac{1-p}{1-\alpha}\right)^{1-\alpha} \right]^n\right).
\end{align*}}

The case for $\setS =\{b\}$ is the same and Theorem~\ref{thm:nonunif-winner-2} follows by maximizing $\delta$ over $\pi\in\Delta$.
\end{proof}

%% file: 05-privacy-GSR.tex

The main result of this section, Theorem~\ref{thm:gsr}, characterizes $(0,\delta,\Delta)$-exact DDP of {\em generalized scoring rules (GSR)} for arbitrary number of candidates, defined below. The main message is that the characterization holds for commonly-used voting rules (Corollary~\ref{cor:gsr-subset}). Therefore, to get the main message, a reader can skip the technical descriptions and definitions below to Corollary~\ref{cor:gsr-subset}.

\begin{definition}[Generalized Scoring Rules (GSR)~\cite{Xia2008}]~\label{def:gsr}
A Generalized Scoring Rule (GSR) is  defined by a number $K\in {\mathbb N}$ and two functions $\funvar{f}:\LinOrd{\mathcal C}\rightarrow\mathbb{R}^{K}$ and $\funvar{g}$, which maps weak orders over the set $\{1,\ldots, K\}$ to $\mathcal C$. 
Given a {vote} $V\in \LinOrd{\mathcal C}$, $\funvar{f}(V)$ is the {\em generalized score vector} of $V$. Given a profile $P$,  we call $\funvar{f}(P) = \sum_{V\in P} \funvar{f}(V)$ the \emph{score}. Then, the winner is given by $\funvar{g}(\funvar{Ord}(\funvar{f}(P)))$, where $\funvar{Ord}$ outputs the weak order of the $K$ components in $\funvar{f}(P)$.
\end{definition}
We say that {a rule is a GSR} if it can be described by some $\funvar{f}$, $\funvar{g}$ as above. Most popular voting rules (i.e., Borda, Plurality, $k$-approval and ranked pairs) are  GSRs. See Example~\ref{example:gsr1} and Example~\ref{example:GSR} for $\funvar{f}$, $\funvar{g}$ for plurality rule and majority rule. The domain of GSRs can be naturally extended to {\em weighted} profiles, where each type of vote is weighted by a real number, due to the linearity of $\funvar{f}$.  

\begin{example}\label{example:gsr1}
The simplest example of a GSR is {\em plurality}. This is the voting rule where each voter chooses exactly one candidate, and the candidate with the most votes is the winner. Here,  $K$ is equal to the number of candidates $m$. Suppose $V$ is a vote (linear order over candidates) where the top candidate is $x_i$. The function $\funvar{f}$ would map $V$ to a vector $\funvar{f}(V) = (0, \cdots, 0, 1, 0, \cdots, 0)$ where the 1 is at position $i$ in the vector. Then, $\funvar{f}(P)$ is exactly the histogram counting the number of times each candidate is ranked at the top of a vote. Finally, the function $\funvar{g}$ chooses the winner.
\end{example}

We now define a set of properties of GSRs to present our characterization of eDDP in Theorem~\ref{thm:gsr}.  

\begin{definition}[Canceling-out, Monotonicity, and Local stability]\label{def:gsr-subset-properties}
	 A voting rule $\mech$ satisfies {\em canceling-out} if for any profile $\vec X$, adding a copy of every ranking does not change the winner. That is, $\mech(\vec X)=\mech(\vec X\cup{\mathcal L(C)})$.
	 
	A voting rule satisfies {\em monotonicity} one cannot prevent a candidate from winning by raising its ranking in a vote while maintaining the order of other candidates. 
	
	A voting rule $\mech$ satisfies {\em local stability} if there exist locally stable profile. A profile $\vec X^*$ is {\em locally stable} (to $\mech$), if there exists a candidate $a$, a ranking $W$, and another ranking $V$ that is obtained from $W$ by raising the position of $a$ without changing the order of other candidates, such that for any $\vec X'$ in the $\gamma$ neighborhood of $\vec X^*$ in terms of $L_{\infty}$ norm, we have (1) $\mech(\vec X')\ne a$, and (2) the winner is $a$ when all $W$ votes in $\vec X'$ becomes $V$ votes. 
\end{definition}

\begin{definition}[Unstable distributions]\label{dfn:unstable} Given a GSR $\mech$, a distribution $\pi$ over $\universe$ is {\em unstable}, if for any $\epsilon>0$, there exists $\vec p$ and $\vec q$ with $\|\vec p\|_2=\|\vec q\|_2<\epsilon$, such that $\mech(\pi+\vec q)\ne \mech(\pi+\vec q)$\footnote{We slightly abuse notation---$\mech(\pi)$  denotes the output of $\mech$ when the voters cast fractional votes according to  $\pi$.}, where $\|\cdot\|_2$ is the $\ell_2$-norm.
\end{definition}

\begin{theorem}[Dichotomy of  Exact DDP for GSR, full proof in Appendix~\ref{fullproof_theorem4}]\label{thm:gsr}
Fix $m\ge 2$ and $\Delta \subseteq \Pi({\mathcal L( \mathcal C)})$ with $|\Delta|<\infty$. For any $n$, any GSR $\mech$ that satisfies monotonicity,  local stability, and canceling-out is $(0, \delta, \Delta)$-DDP, where $\delta$ is 
\begin{itemize}
    \item $\Theta(\sqrt{1/n})$, if $\Delta$ contains the uniform distribution over $\mathcal L(\mathcal C)$, or
    \item $\exp(-\Omega(n))$, if $\Delta$ does not contain any unstable distribution. 
\end{itemize}  
\end{theorem}

\begin{proof}
[Proof sketch for Theorem~\ref{thm:gsr}] (See Appendix~\ref{fullproof_theorem4} for the full proof)

We first prove the $\delta = \exp[-\Omega(n)]$ case. Recalling the proof of Theorem~\ref{thm:nonunif-winner-2}, we know that $\delta$ is closely related to the probability of $\ed(a)$ for some $a\in\mc$. 
It turns out that this is also the case for any GSR $\mech$ that also satisfies monotonicity. Applying our trails technique, we have
{\footnotesize $$\delta \leq \max_a \sum_{P\in\ed(a)}\Pr(P-V),$$}
where $V$ is a vote s.t.~there exists vote $W$ with $\mech(P-V+W)\neq a$. Thus, we know $\delta$ is upper bounded by the probability of all profiles ($P-V$) ``close'' to a tie of voting rule $r$. For any unstable distribution $\pi$, we can prove that the center of $\pi$ is reasonably ``far'' away from any profile in $\ed(a)$ (or ``far'' away from any ties). Then, 
the exponential upper bound follows after Chernoff bound and union bound. The proof for this part is similar to the analysis of probabilities of tied profiles as in~\cite{Xia2008}.

We now move on to the $\delta = \Theta(\sqrt {1/n})$ case. The upper bound $O(\sqrt{1/n})$ also derived from the trails technique's result: $\delta \leq \max_a \sum_{P\in\ed(a)}\Pr(P-V)$. General framework of our proof is similar with the $\delta = \exp[-\Omega(n)]$ case. Since adding any vote to a uniform profile results in a new winner, we know the uniform distribution of preferences is always an unstable distribution when  requirements in Theorem~\ref{thm:gsr} are met. Thus, we can prove that the center of the profiles' distribution (multinomial distribution in $m!$-dimensional space) is ``close'' to a tie. Then, we apply Stirling's formula to each trails and give an upper bounds to $\Pr(P-V)$ for profiles $P\in \ed(a)$.

For the lower bound $\Omega(\sqrt{1/n})$, canceling-out and locally stability are used to construct a ``good'' subset of profiles. At a high level, canceling-out ensures that the constructed subset is large enough, and locally stability ensures the trails constructed from the selected subset is long enough. Our subset is contracted by certain profiles with $O(\sqrt n)$ distance\footnote{we use $\ell_2$ distance in the $m!$-dimensional space of profile.} from the center of profile distribution in the direction of local stable profile. Giving a lower bound to the $\Pr(P-V)$ for any profile $P$ in our selected subset is the most non-trivial part of this proof and is quite different from the proof in~\cite{Xia2008}. Unlike the profiles $P$ in our selected subset of profiles, $P-V$ do not necessarily concentrated in a specific region in the space of profiles. Here, we use a non-i.i.d.~version of Lindeberg-Levy central limit theorem~\cite{greene2003econometric} to analyze the multinomial distribution of $m!$ kinds of votes.
\end{proof}

Next, we use a simple example of  majority rule to show the results in Theorem~\ref{thm:gsr} matches the 2-candidate results in Section~\ref{section:exactDDP}. In the following example, we also provide the intuitions on how to describe voting rules in the language of GSR.

\begin{example}[Example of Definition~\ref{def:gsr} and Theorem~\ref{thm:gsr}]\label{example:GSR}  
Let $\universe = \rangeM= \mc = \{c_1,c_2\}$, $V = [c_1\succ c_2]$, and $W = [c_2 \succ c_1]$. For the \emph{majority rule} with $\alpha=0.5$, we have $\funvar{f}(V) = (1,0)$ and $\funvar{f}(W) = (0,1)$. Then, the winner is chosen according to $\funvar{g}$ corresponding to the largest component in $\funvar{f}(P)$. Recalling our definition of unstable distribution, we know $(\frac 1 2,\,\frac 1 2)$ is the only unstable distribution for 2-candidate majority rule. This is the intuitive reason behind $\delta = \Theta(\sqrt{1/n})$ when $\pi = (\frac 1 2,\,\frac 1 2)$ for both Theorem~\ref{thm:gsr} and Theorem~\ref{thm:nonunif-winner-2} (when $\alpha = 0.5$). For any other $\pi \neq (\frac 1 2,\,\frac 1 2)$, these two theorems result in $\delta = \exp[-\Omega(n)]$. We note that while Theorem~\ref{thm:gsr} covers more voting rules,  Theorem~\ref{thm:nonunif-winner-2} is a more fine-grained result  for two candidates.

\end{example}

\begin{corollary}\label{cor:gsr-subset} Plurality, veto, $k$-approval, Borda, Maximin, Copeland, Bucklin, Ranked Pairs, Schulze (see e.g.~\cite{Xia2008}) 
are $\left(0,\Theta\left(1/\sqrt {n}\right),\Delta\right)$-\eDDP~ when $\Delta$ contains the uniform distribution.
\end{corollary}
\begin{proof} As shown in Definition~\ref{def:gsr-subset-properties}, {\em canceling-out} and {\em monotonicity} are very natural properties of most voting rules. These two properties can be easily checked according to the definitions of voting rules discussed in Corollary~\ref{cor:gsr-subset}. In the next proposition, we prove a more generalized version of Corollary~\ref{cor:gsr-subset} for {\em local stability}, which indicate a large subset of the voting rules can satisfies all properties required by Theorem~\ref{thm:gsr}.
\begin{proposition}\label{prop:gsr-subset-positional1} All positional scoring rules and all Condorcet consistent and monotonic rules satisfy the property of {\em local stability}.
\end{proposition}

\begin{proof} 
Let $s_i$ to denote the score of the $i$-th candidate ($f(P)$ in definition~\ref{def:gsr}). Suppose $s_1=\cdots=s_l>s_{l+1}$. We let $V=[a\succ c_1\succ c_{l-1}\succ b\succ \text{others}]$ and $W=[c_1\succ c_{l-1}\succ b\succ a\succ \text{others}]$. Let $M$ be the permutation $c_1\rightarrow c_2\rightarrow\ldots c_{m-2}\rightarrow c_1$. Let $V_1=[a\succ b\succ \text{others}]$  and $V_2=[b\succ a\succ \text{others}]$. Let $P'=\bigcup_{i=1}^{m-2}M^i(V_1)\cup M^i(V_2)$. Let  $P^*=2P'\cup\{V,W\}$. It follows that $a$ and $b$ are the only two candidates tied in the first place in $P^*$. Therefore, there exists $\epsilon$ to satisfy the condition in local stability.

The same profile can be used to prove the local stability of all Condorcet consistent and monotonic rules.
\end{proof}
Then, Corollary~\ref{cor:gsr-subset} follows by combining the results for all three properties.
\end{proof}
Another commonly-used GSR called STV does not satisfy monotonicity, which means that Theorem~\ref{thm:gsr} does not apply. However, empirical results (Section~\ref{section:empirical-results}) suggest that STV is likely also $\left(0,\Theta\left(1/\sqrt{n}\right),\Delta\right)$-\eDDP~for this distribution.

\section{CONCRETE ESTIMATION OF THE PRIVACY PARAMETERS}\label{section:empirical-results-main}
We present an example of computing concrete estimates of $(0, \delta, \Delta)$-exact DDP values for several GSRs. For this example, we let $\Delta = \{\pi\}$ such that $\pi\in \Pi(\{x_1, x_2, x_3\})$ and $\pi(x_i) = \pi(x_j) = 1/3$ (i.e., votes are i.i.d. and uniform).  
We generated these concrete estimates via exhaustive search over possible profiles for 3 candidates and $n \leq 50$ votes, and computing the $\delta$ values exactly for each $n$. Since we know that $\delta = \Theta(1/\sqrt{n})$, we fit these values to $\delta(n) = \frac{1}{\sqrt{an + b}}$ via linear regression. 
We rank voting rules from most to least private. The larger the $a$, the smaller the $\delta$ value and thus more private: 

\vspace{2mm}
\hfil {\bf 2-approval $\rhd$ Plurality $\rhd$  Maximin  $\rhd$  STV  $\rhd$  Borda }\hfill
\vspace{2mm}

We showed in Table~\ref{GSR-concrete-delta} (Section~\ref{section:introduction}, also see Table~\ref{GSR-concrete-delta-big} in Appendix~\ref{section:empirical-results} for more information) the fitted $\delta$ curves. Figure~\ref{fig:GSR-winner-score-small} shows the comparison between Plurality, Borda, and STV voting rules w.r.t.~their $\delta$ values in $(0,\delta,\Delta)$-eDDP, when fitted to $\delta(n) = \frac{1}{\sqrt{an + b}}$.

\begin{figure}[ht]
	\centering
	\includegraphics[width=0.47\textwidth]{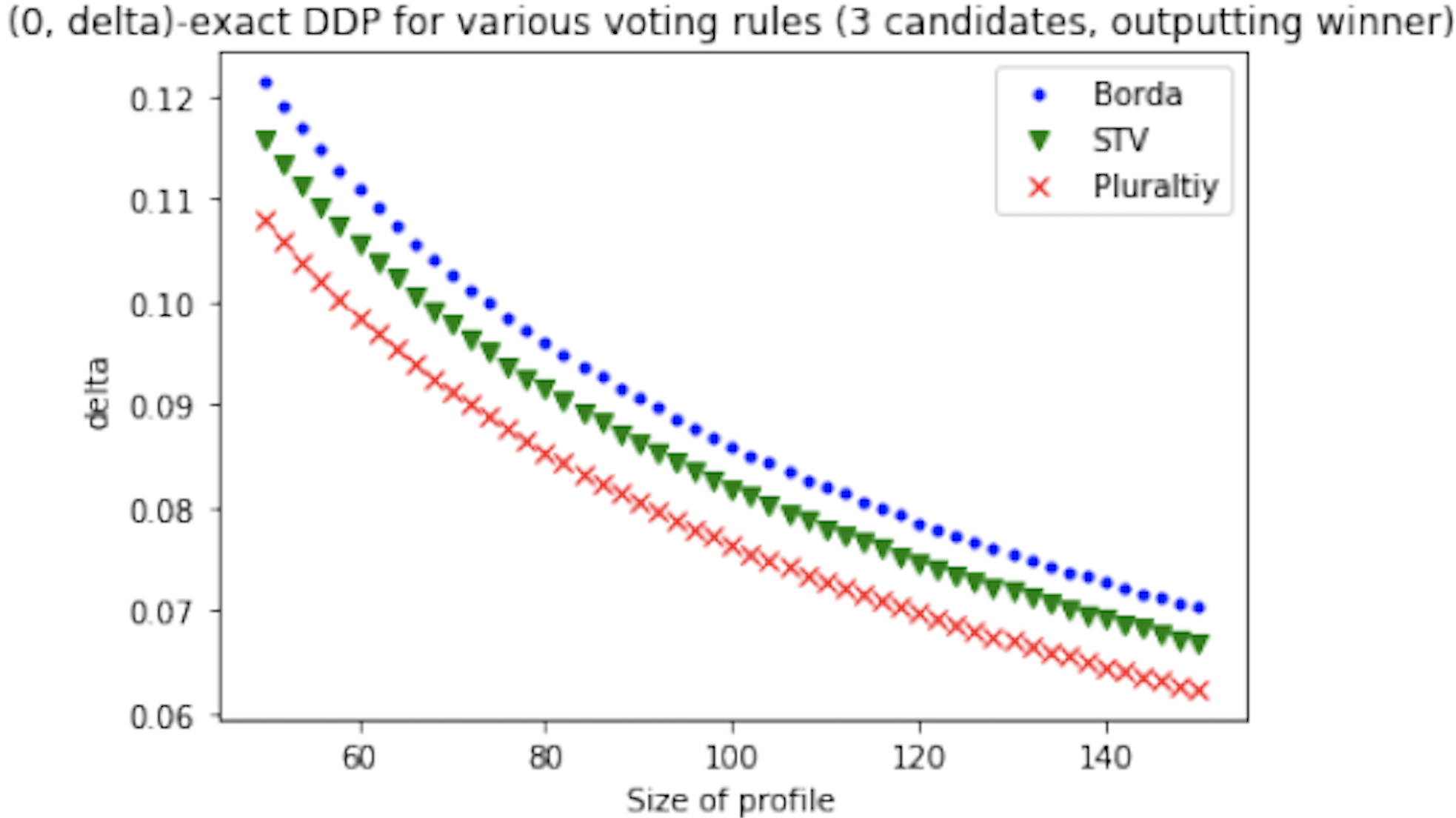}
	\caption{{\small The $\delta$ values in $(0, \delta, \Delta)$-eDDP for Borda, STV, and plurality in our concrete estimates.}}
	\label{fig:GSR-winner-score-small}
\end{figure}

\section{SUMMARY AND FUTURE WORK}
We address the limitation of DP in deterministic voting rules by introducing and characterizing (exact) DDP for voting rules, leading to an encouraging message about the good privacy of commonly-studied voting rules and a framework to compare them w.r.t.~eDDP. There are many directions for future work. An immediate open question for theoretical study is to extend our studies to general $(\epsilon$, $\delta)$, and non-i.i.d.~distributions, as well as to other high-stakes social choice procedures such as matching and resource allocation.  On the practical side, it could be informative to study the eDDP of  other data that is often published during an election, such as demographic information, and  interpret their consequences.

%% file: app-related_literature.tex
\section{Additional Related Literature}\label{appendix:related_literature}

The first works on DP described how one can create mechanisms for answering standard statistical queries on a database (e.g., number of records with some property or histograms) in a way that satisfies the DP definition. This ignited a vast and rapidly evolving line of research on extending the set of mechanisms and achieving different DP guarantees---we refer the reader to~\cite{Dwork2014} for an (already outdated) survey---to a rich literature of relaxations to the definition, e.g., ~\cite{Bhaskar2011,Leung2012a,Duan2009,Bassily2013}, that capture among others, noiseless versions of privacy, as well as works studying the trade-offs between privacy and utility of various mechanisms~\cite{McSherry2007,Blum2008,Hardt2010,Bassily2015,Ghosh2009}.

\emph{Generalized Scoring Rules} (GSRs) is a  class of voting rules that include many commonly studied voting rules, such as Plurality, Borda, Copeland, Maximin, and STV~\cite{Xia2008}. It has been shown that for any GSR the probability for a group of manipulators to be able to change the winner has a phase transition~[Xia and Conitzer, 2009; Mossel {\em et al}., 2013]. An axiomatic characterization of GSRs is given in~\cite{Xia2009}. The most robust GSR with respect to a large class of statistical models has been characterized~\cite{Caragiannis14:Modal}. Recently GSRs have been extended to an arbitrary decision space, for example to choose a set of winners or rankings over candidates~\cite{Xia2015}.

\paragraph{Relaxations to Differential Privacy and Noiseless Functions} Relaxations to differential privacy have been proposed to allow functions with less to no noise to achieve a DP-style notion of privacy. Kasiviswanathan and Smith~\cite{Kasiviswanathan2008} formally proved that differential privacy holds in presence of arbitrary adversarial information, and formulated a Bayesian definition of differential privacy which makes adversarial information explicit. Hall et al.~\cite{Hall2012} suggested adding noise to only certain values (such as low-count components in histograms) to achieve a relaxed notion of Random Differential Privacy with higher accuracy. Taking advantage of (assumed) inherent randomness in the database, several works have also put forward DP-style definitions which allow for noiseless functions. Duan~\cite{Duan2009} showed that sum queries of databases with i.i.d. rows can be outputted without noise. Bhaskar et al.~\cite{Bhaskar2011} introduced Noiseless Privacy for database distributions with i.i.d. rows, whose parameters depend on how far the query is from a function which only depends on a subset of the database. Motivated by Bayesian mechanism design, Leung and Lui~\cite{Leung2012}, suggested noiseless sum queries and introduced Bayesian differential privacy for database distributions with independent rows, where the auxiliary information is some number of revealed rows. 

These ideas were generalized and extended by Bassily et al. who introduced {\em distributional differential privacy (DDP)}~\cite{Bassily2013,Groce2014}. Informally, given a distribution $(X, Z)$, where $X$ is the adversary's uncertainty in the database distribution and $Z$ is a parameter used for proving composition theorems (i.e. computing DDP when outputting results from two functions that are both DDP with some parameters), we say a function $\mech$ is $(\epsilon,\delta, \Delta = \{(X, Z)\})$-DDP if its output distribution $\mech(X) |Z$ can be simulated by a simulator that is given the database missing one row. In these works, noiseless functions which have been shown to satisfy DDP are exact sums, truncated histograms, and \emph{stable} functions where with large probability, the output is the same given neighboring databases.

%% file: app-appendix.tex

\section{Distributional Differential Privacy for Voting}\label{appendix:Theorem3}

\subsection{Equivalence of our DDP definition and that of Bassily et. al 2013}
\label{sec:equivalence}
For completeness we present the DDP definition of~\cite{Bassily2013}. However, this definition is harder to work with, and harder to explain conceptually. Thus, we choose to present Definition~\ref{def:our-DDP}. Below, we will show that in our setting, these two definitions are equivalent.

\begin{definition}[Distributional Differential Privacy (DDP)~\cite{Bassily2013}]~\label{def:orig-DDP}
	A function $\mech:\universe^{*}\rightarrow\rangeM$ is $(\epsilon, \delta, \Delta)$-distributional differentially private (DDP) if there is a simulator $\Sim$ such that for all $D = (\pi, Z)\in \Delta$, $\vec X\sim \pi$, for all $i$, $(x, z)\in \Supp{X_{i}, Z}$ (where $X_{i}$ denotes the random variable that is the $i$th component of $\vec X$, $X_{-i}$ denotes the r.v. that is $\vec X$ without the $i$th component, and $\Supp{.}$ denotes the \emph{support} of a distribution), and all sets $\setS\subseteq\rangeM$,
	$$
	\Pr_{\vec X\sim \pi}(\mech(\vec X)\in \setS |X_{i} = x, Z = z) \leq e^{\epsilon} \Pr_{\vec X\sim \pi}(\Sim(X_{-i})\in \setvar{S}|X_{i} = x, Z = z) + \delta
	$$
	and
	$$
	\Pr_{\vec X\sim \pi}(\Sim(X_{-i})\in \setS|X_{i} = x, Z = z) \leq e^{\epsilon} \Pr_{\vec X\sim \pi}(\mech(\vec X)\in \setvar{S} |X_{i} = x, Z = z) + \delta
	$$
\end{definition}

\begin{lemma}[Equivalence of definitions]\label{cor:equivalence-of-defs}
	For any $\universe$, let $\Delta\subseteq \Pi(\universe)$ and $\Delta' = (\Delta, Z = \emptyset)$ (where $Z$ is a parameter in the~\cite{Bassily2013} definition). Suppose $\mech$ is $(\epsilon, \delta, \Delta')$-(simulation-based) DDP~\cite{Bassily2013}, then $\mech$ is $(2\epsilon, (1+e^{\epsilon})\delta, \Delta)$-DDP for our Definition~\ref{def:our-DDP}. Conversely, if $\mech$ is $(\epsilon, \delta, \Delta)$-DDP  for Definition~\ref{def:our-DDP} then $\mech$ satisfies $(\epsilon, \delta, \Delta')$-(simulation-based) DDP.
\end{lemma}
\begin{proof}[Lemma~\ref{cor:equivalence-of-defs}]

	We prove the first statement, that is, if $\mech$ is $(\epsilon, \delta, \Delta')$-(simulation-based) DDP~\cite{Bassily2013}, then $\mech$ is $(2\epsilon, (1+e^{\epsilon})\delta, \Delta)$-DDP of Definition~\ref{def:our-DDP}. \\
	
	By the definition of $\mech$ being $(\epsilon, \delta, \Delta')$-(simulation-based) DDP, the simulator $\Sim$ has to satisfy the below inequalities for any $(\pi, Z) \in \Delta'$, any $i$, and $x \in \Supp{X_{i}}$ (for $\vec X \sim \pi$). With  $Z = \emptyset$, we can write the inequalities in the DDP definition without $Z$ as
	$$\Pr_{\vec X\sim \pi}(\mech(\vec X)\in \setS\ |\ X_{i} = x) \leq e^{\epsilon}\Pr_{\vec X\sim \pi}(\Sim(X_{-i})\in \setS\ |\ X_{i} = x) + \delta$$
	\begin{align}\label{eqn:DDP-without-sim}
	\Pr_{\vec X\sim \pi}(\Sim(X_{-i})\in \setS | X_{i} = x) \leq e^{\epsilon}  \Pr_{\vec X\sim \pi}(\mech(\vec X)\in \setS | X_{i} = x) + \delta
	\end{align}
	(We make $\vec X\sim \pi$ implicit to ease presentation.) Now consider any $x'\in \Supp{X_{i}}$, possibly different from the $x$ above. By the definition of DDP, the inequalities should also hold for $x'$, i.e.
	$$\Pr(\mech(\vec X) \in \setS\ |\ X_{i} = x') \leq e^{\epsilon}\Pr(\Sim(X_{-i})\ |\ X_{i} = x') + \delta $$
	Since the simulator is not given $i$th entry of the database, its output does not depend on the value of the $i$th row. Moreover, if database rows are independent, the distributions $X_{-i} | X_{i} = x' = X_{-i}| X_{i} = x$. Thus $\Pr(\Sim(X_{-i})\ |\ X_{i} = x')  =  \Pr(\Sim(X_{-i})\in \setS\ |\ X_{i} = x)$. So,
	\begin{align*}
	\Pr(\mech(\vec X)\in \setS\ |\ X_{i} = x') &\leq e^{\epsilon}\Pr(\Sim(X_{-i})\in \setS\ |\ X_{i} = x) + \delta\\
	\Pr(\mech(\vec X)\in \setS\ |\ X_{i} = x') &\leq e^{\epsilon}(e^{\epsilon}\Pr(\mech(\vec X)\in \setS | X_{i} = x) + \delta)+ \delta \text{ (By Equation~\ref{eqn:DDP-without-sim} above.)}\\
	\Pr(\mech(\vec X)\in \setS\ |\ X_{i} = x') & \leq e^{2\epsilon}\Pr(\mech(X)\in \setS\ |\ X_{i} = x)  + e^{\epsilon}\delta + \delta
	\end{align*}
	Thus, we have shown that for all $x, x'\in \Supp{X_{i}}$ (and all $i$), $$\Pr(\mech(X)\in \setS\ |\ X_{i} = x') \leq e^{2\epsilon}\Pr(\mech(X)\in \setS\ |\ X_{i} = x)  + (e^{\epsilon} + 1)\delta$$ So,  $\mech$ is $(2\epsilon, (1+e^{\epsilon})\delta, \Delta)$-DDP, proving the first statement.\\
	
	We now prove the second statement. That is,  if $\mech$ is $(\epsilon, \delta, \Delta)$-DDP of Definition~\ref{def:our-DDP} then $\mech$ is $(\epsilon, \delta, \Delta')$-(simulation-based) DDP. To do so, we define the simulator $\Sim$ to be the algorithm which inserts any $x'\in \Supp{X_{i}}$ to the missing $i$th row of the database, and apply $\mech$ to the result. By independence of rows, $\Pr(\Sim(X_{-i})\ |\ X_{i} = x)	= \Pr(\Sim(X_{-i})\ |\ X_{i} = x')$ by our definition of $\Sim$, equal to $\Pr(\mech(X)\ |\ X_{i} = x')$. Then, for any $X \in \Delta$, $i$, and $x, x'\in \Supp{X_{i}}$,	
	\begin{align*}
	\Pr(\Sim(X_{-i})\in \setS\ |\ X_{i} = x)	&= \Pr(\mech(X)\in \setS\ |\ X_{i} = x') \leq e^{\epsilon}\Pr(\mech(X)\in \setS\ |\ X_{i} = x) + \delta
	\end{align*}
	by inequality of Definition~\ref{def:our-DDP}. This proves the second statement.
\end{proof}

\subsection{Proof of Theorem~\ref{thm:epsdelta-histogram-c}: DDP of Histogram}\label{appendix:thm1}

\paragraph{Theorem~\ref{thm:epsdelta-histogram-c} (DDP of $\Hist$)} {\em Given any $\universe=\{x_1,\ldots,x_{\bins}\}$ and $\Delta\subseteq \Pi(\universe)$  with $|\Delta|<\infty$, let $p_{\min} = \min_{\pi\in\Delta, i\le \bins} (\pi(x_i))$. For any $n\in\Natural$ and any $\epsilon \geq 2 \ln(1+\frac{1}{p_{\min}n})$, $\Hist$ for $n$ voters is $(\epsilon,\delta,\Delta)$-DDP where $\delta =\exp(-\Omega(np_{\min}[\min(2\ln(2), \epsilon)]^{2}))$.}

\begin{proof}At a high level, the proof is similar to Theorem 8 of~\cite{Leung2012}.
	
	Fix $\pi \in \Delta$.	Since votes are i.i.d. and all $i\in [n]$ are equivalent, we simplify $\Pr_{\vec X\sim \pi}(\Hist(\vec X) \in \setS | X_{i} = x) $  as $ \Pr(\Hist(x, X_{-1}) \in \setS)$, where $X_{-1}$ refers to $\vec X$ without the first vote.
	
We need to show that for all $x_{i}, x_{j} \in \{x_{1}, \cdots, x_{\bins}\}$, and all $\setS \subseteq  \Natural^{\bins}$:
	\begin{align*}
	\Pr(\Hist(x_{i}, X_{-1}) \in \setS) \leq e^{\epsilon}\Pr(\Hist(x_{j}, X_{-1}) \in \setS) + \delta
	\end{align*}
	
	We observe that for any set $\setvar{B}$ and $x$:
	\begin{align}
	\Pr(\Hist(x, X_{-1}) \in \setS) &= \Pr(\Hist(x, X_{-1}) \in \setS\cap \overline{\setvar{B}}) + \Pr(\Hist(x, X_{-1}) \in \setS\cap \setvar{B})\\
	&\leq \Pr(\Hist(x, X_{-1}) \in \setS\cap \overline{\setvar{B}}) + \Pr(\Hist(x, X_{-1}) \in \setvar{B})\label{eqn:nonunif-histogram-c-1}
	\end{align}
	
	Let $\setvar{B}$ be the set of all histogram $t\in  \Natural^{\bins}$ where $t_{i} > p_{i}(n-1)e^{\epsilon/2}$ and $t_{j} < p_{j}(n-1)e^{-\epsilon / 2}$. Fix a choice of $\epsilon > 2 \ln(1+\frac{1}{p_{\min} n})$. We claim that for $\delta = \exp(\Omega(np_{\min}(\min(2\ln(2), \epsilon))^{2})$, the following hold:\\
	
	\textit{Claim 1:} $\Pr(\Hist(x_{i}, X_{-1})  \in \setS\cap \overline{\setvar{B}}) \leq e^{\epsilon} \Pr(\Hist(x_{j}, X_{-1})  \in \setS\cap \overline{\setvar{B}})$\\
	
	\textit{Claim 2:} $\Pr(\Hist(x_{i}, X_{-1}) \in \setvar{B}) \leq \delta$\\
	
	If both claims are true, then by Inequality (\ref{eqn:nonunif-histogram-c-1}),
	\begin{align*}
	\Pr(\Hist(x_{i}, X_{-1}) \in \setS)	&\leq \Pr(\Hist(x_{i}, X_{-1}) \in \setS\cap \overline{\setvar{B}}) + \Pr(\Hist(x_{i}, X_{-1}) \in \setvar{B})\\
	&\leq e^{\epsilon} \Pr(\Hist(x_{j}, X_{-1}) \in \setS\cap \overline{\setvar{B}}) + \delta\\
	&\leq e^{\epsilon} \Pr(\Hist(x_{j}, X_{-1}) \in \setS)+ \delta
	\end{align*}
	which proves the theorem. Below we will prove both claims.\\
	
	\textit{Claim 1 proof:}\\
	Since all entries in random variable $X_{-1}$ are i.i.d., the random variable \\ $\Hist(X_{-1})$ which outputs the histogram of the database has distribution equal to the multinomial distribution on $n-1$ trials and $\bins$ events:
	\begin{equation*}
	\Pr(\Hist(X_{-1}) =  (t_{1}, \cdots, t_{\bins})) = \frac{(n-1)!}{t_{1}!\cdots t_{\bins}!}\ p_{1}^{t_{1}}\cdots p_{\bins}^{t_{\bins}}
	\end{equation*}
	where $t_{i}$ is the count of entries with the value $x_{i}$ and $p_{i}$ is the probability for an entry to have the value $x_{i}$. \\
	
	Thus, 
	\begin{equation*}
	\Pr(\Hist(x_{i}, X_{-1}) =  (t_{1}, \cdots, t_{\bins})) = \dfrac{(n-1)!}{t_{1}!\cdots (t_{i}-1)!\cdots t_{\bins}!} p_{1}^{t_{1}}\cdots p_{i}^{t_{i} - 1}\cdots p_{\bins}^{t_{\bins}}
	\end{equation*}
	and
	\begin{equation*}
	\Pr(\Hist(x_{j}, X_{-1}) = (t_{1}, \cdots, t_{\bins})) = \dfrac{(n-1)!}{t_{1}!\cdots (t_{j}-1)!\cdots t_{\bins}!} p_{1}^{t_{1}}\cdots p_{i}^{t_{j} - 1}\cdots p_{\bins}^{t_{\bins}}
	\end{equation*}
	So, for every $t = (t_{1}, \cdots, t_{\bins})\in \setS \cap \overline{\setvar{B}}$:
	\begin{align*}
	\dfrac{Pr(\Hist(x_{i}, X_{-1}) = t)}{Pr(\Hist(x_{j}, X - 1) = t)} & = \dfrac{ \dfrac{(n-1)!}{t_{1}!\cdots (t_{i}-1)!\cdots t_{\bins}!} p_{1}^{t_{1}}\cdots p_{i}^{t_{i} - 1}\cdots p_{\bins}^{t_{\bins}} }{ \dfrac{(n-1)!}{t_{1}!\cdots (t_{j}-1)!\cdots t_{\bins}!} p_{1}^{t_{1}}\cdots p_{i}^{t_{j} - 1}\cdots p_{\bins}^{t_{\bins}} }\\
	& = \dfrac{t_{i}}{p_{i}} \dfrac{p_{j}}{t_{j}}\\
	& = \dfrac{t_{i}}{p_{i}(n-1)} \dfrac{p_{j}(n-1)}{t_{j}}\\
	&\tag*{By definition of $\setvar{B}$, $t_{i} > p_{i}(n-1) e^{\epsilon/2}$ or $t_{j} < p_{j}(n-1) e^{-\epsilon/2}$,}\\
	&\tag*{so $t\in \overline{\setvar{B}}$ has $t_{i} \leq t_{i}(n-1)e^{\epsilon/2}$ and $t_{j} \geq p_{j}(n-1) e^{-\epsilon/2}$}\\
	& \leq \dfrac{p_{i}(n-1)e^{\epsilon/2}}{p_{i}(n-1)} \dfrac{p_{j}(n-1)}{p_{j}(n-1) e^{-\epsilon/2}} = e^{\epsilon/2}\times e^{\epsilon/2} = e^{\epsilon}
	\end{align*}
	This proves Claim 1.\\
	
	\textit{Claim 2 proof:}
	Recall $\setvar{B}$ is the set of all histogram $t\in  \Natural^{\bins}$ where $t_{i} > p_{i}(n-1)e^{\epsilon/2}$ and $t_{j} < p_{j}(n-1)e^{-\epsilon / 2}$. For any $i\in \{1, \cdots, \bins\}$ let $\Hist(x, X_{-1})_{i}$ denote $i$th component of the random variable $\Hist(x, X_{-1})$.
	\begin{align*}
	\Pr&(\Hist(x_{i}, X_{-1}) \in \setvar{B}) \\
	&= \Pr\left(\Hist(x_{i}, X_{-1})_{i} > p_{i}(n-1)e^{\epsilon/2} \text{ or } \Hist(x_{i}, X_{-1})_{j} < p_{j}(n-1)e^{-\epsilon/2}\right)\\
	& \leq \Pr\left(\Hist(x_{i}, X_{-1})_{i} > p_{i}(n-1)e^{\epsilon/2}\right) + \Pr\left(\Hist(x_{i}, X_{-1})_{j} < p_{j}(n-1)e^{-\epsilon/2}\right) \tag{By union bound}\\
	& = \Pr\left(1 + \funvar{Bin}(n-1, p_{i}) > p_{i}(n-1)e^{\epsilon/2}\right) + Pr\left(\funvar{Bin}(n-1, p_{j}) < p_{j}(n-1)e^{-\epsilon/2}\right) \\
	& \tag{Where $\funvar{Bin}(n, p)$ denotes  binomial r.v. with $n$ trials and success probability $p$}\\
	& = \Pr(\funvar{Bin}(n-1, p_{i}) >  p_{i}(n-1)(e^{\epsilon/2} - (p_{i}(n-1))^{-1})) \\
	& \tab + \Pr(\funvar{Bin}(n-1, p_{j}) < p_{j}(n-1)e^{-\epsilon/2} )
	\end{align*}
	
	The random variable $\funvar{Bin}(n-1, p_{i})$ has mean $\mu = p_{i}(n-1)$. When
	$$
	2\ln(1+\frac{1}{p_{i}(n-1)}) < 2\ln(1+\frac{1}{p_{\min}(n-1)}) < \epsilon \leq 2\ln(2) < 2\ln(2 + \frac{1}{p_{i}(n-1)}) 
	$$ 
	we have $0 < \beta = e^{\epsilon/2} - (p_{i}(n-1))^{-1} - 1 < 1$. By Chernoff bound,
	
	\begin{align*}
	\Pr(\funvar{Bin}(n-1, p_{i})  > (1+ \beta)\mu &\leq e^{-\mu \beta^{2} / 3}\\
	&= \exp(-\Omega(p_{i}(n-1)(e^{\epsilon/2} - (p_{i}(n-1))^{-1} - 1)^{2}))\\
	&= \exp(-\Omega(p_{i}n\epsilon^{2}))
	\end{align*}
	
	The random variable $\funvar{Bin}(n-1, p_{j})$ has mean $\mu = p_{j}(n-1)$. By Chernoff bound, for any $0 < \beta = 1 - e^{-\epsilon/2} < 1$ (ie.  $\epsilon > 0$),
	
	\begin{align*}
	Pr(\funvar{Bin}(n-1, p_{j})  < (1-\beta) \mu) &\leq e^{-\mu \beta^{2} / 2}\\
	&= \exp(-\Omega(p_{j}(n-1)(1 - e^{-\epsilon/2})^{2}))\\
	&= \exp(-\Omega(p_{j}n\epsilon^{2}))
	\end{align*}
	
	So that:
	\begin{align*}
	\Pr(\Hist(x_{i}, X_{-1}) \in \setvar{B}) &\leq \Pr(\funvar{Bin}(n-1, p_{i}) >  p_{i}(n-1)(e^{\epsilon/2} - (p_{i}(n-1))^{-1})) \\
	& \tab + \Pr(\funvar{Bin}(n-1, p_{j}) < p_{j}(n-1)e^{-\epsilon/2} )\\
	&\leq \exp(-\Omega(p_{i}n\epsilon^{2})) + \exp(-\Omega(p_{j}n\epsilon^{2}))\\
	&\leq \exp(-\Omega(p_{\min}n\epsilon^{2})) = \delta
	\end{align*}
	
	for $2\ln(1+\frac{1}{p_{\min}(n-1)}) < \epsilon \leq 2\ln(2)$. To get rid of the upper bound on $\epsilon$, notice when $\epsilon = 2\ln(2)$, $\delta = \exp(-\Omega(p_{\min}n(2\ln(2))^{2}))$ suffices to satisfy the inequality
	{\small\begin{align*}
	\Pr(\Hist(x_{i}, X_{-1}) \in \setS) \leq e^{\epsilon}\Pr(\Hist(x_{j}, X_{-1}) \in \setS) + \delta
	\end{align*}}
	Thus, when $\epsilon > 2\ln(2)$, the same $\delta = \exp=(\Omega(np_{\min}[\min(2\ln(2), \epsilon)]^{2})) = \exp(-\Omega(p_{\min}n(2\ln(2))^{2}))$ also suffices, as a larger $\epsilon$ only makes the right hand side of the inequality larger.\\
	
	This proves Claim 2.

\end{proof}


The definition of distributional differential privacy, like differential privacy, is immune to post-processing. This means that if $\mech$ is $(\epsilon, \delta, \Delta)$-DDP, and $\funvar{f}$ is a function on the output of $\mech$, then $\funvar{f} \circ \mech$ (their composition) is also $(\epsilon, \delta, \Delta)$-DDP. Note that post-processing immunity is not a property of exact privacy, since exact privacy describes tight bounds on $\epsilon, \delta$.

\begin{lemma}[Immunity to Post-processing]\label{lem:foM}
	Suppose  $\mech:\universe^{*}\rightarrow\rangeM$  is $(\epsilon, \delta, \Delta)$-\altDDP. Let $\funvar{f}:\rangeM\rightarrow\rangeM'$ be a deterministic function. Then $\funvar{f}\circ\mech:\universe^{*}\rightarrow\rangeM'$ is also $(\epsilon, \delta, \Delta)$-\altDDP.
\end{lemma}
\begin{proof}
	For any $\pi\in \Delta$, $x, x'\in \Supp{X_{i}}$ and $\setvar{S}\subseteq \rangeM'$, let $\setvar{W} = \{w\ :\ \funvar{f}(w) \in \setvar{S}\}$. Then
	\begin{align*}
	&\Pr_{\vec X \sim \pi}(\funvar{f}(\mech(\vec X)) \in \setvar{S}\ |\ X_{i} = x)\\
	&= \Pr(\mech(\vec X)\in \setvar{W} \ |\ X_{i} = x) \tag{By definition of $\setvar{W}$}\\
	&\leq e^{\epsilon} \Pr(\mech(\vec X)\in \setvar{W} \ |\ X_{i} = x') + \delta \tag{By $\mech$ being $(\epsilon, \delta, \Delta)$-DDP}\\
	&= e^{\epsilon} \Pr(\funvar{f}(\mech(\vec X)) \in \setvar{S}\ |\ X_{i} = x) + \delta \tag{By definition of $\setvar{W}$}
	\end{align*}
\end{proof}

By post-processing immunity, the parameters proven in Theorem~\ref{thm:epsdelta-histogram-c} also apply to functions whose outputs are based on the histogram of the database, such as most voting rules.

\section{Exact Privacy of Voting Rules: Two Candidate-Case (Cont'd)}

\subsection{Proof of Lemma~\ref{lem:trail}: Trails Technique}\label{sec:trailproof}

\paragraph{Lemma~\ref{lem:trail} (Trails)}
{\em	Let $\mt$ be a trail with direction $(j,k)$, and let $\pi$ be a distribution where votes are independently distributed. For any $i$, $x_{j}, x_{k} \in \Supp{X_{i}}$,
	\begin{equation}\nonumber
	\begin{split}
	&\Pr_{\vec X\sim \pi}(\Hist(\vec X) \in \mt\ |\ X_{i} = x_{j}) -	\Pr_{\vec X\sim \pi}(\Hist(\vec X) \in \mt\ |\ X_{i} = x_{k}) \\
	=\;&\Pr_{\vec X\sim \pi}(\Hist(\vec X) = \ext{\mt}\ |\ X_{i} = x_{j}) -  \Pr_{\vec X\sim \pi}(\Hist(\vec X) =\ent{\mt}\ |\ X_{i} = x_{k})
	\end{split}
	\end{equation}
}

\begin{proof} [Proof for Lemma~\ref{lem:trail}]
    Fix distribution $\pi$ over $n$ votes, where each vote is independently distributed. For $\vec X\sim\pi$, denote $X_{-i}$ as the random variable $\vec X$ but without the $i$th vote.
	The equality in the lemma comes from the simple observation that when votes are independently distributed, for any histogram $t\in \Natural^{\bins}$ and any $j\in [\bins]$
	$$
	\Pr_{\vec X \sim \pi}(\Hist(\vec X) = t|X_{i} = x_{j}) = \Pr_{\vec X \sim \pi}(\Hist(X_{-i}) = t - x_{j})
	$$
	(Below, $\vec X \sim \pi$ is implicit). Let $q$ be the length of the trail. For any $0\leq z<q$, let $t_{z} = \ent{\mt}-zx_{j} + zx_{k}$. Then,
\begin{align*}
	&\Pr(\Hist(\vec X) = t_{z}|X_{i} = x_{j}) \\
	&= \Pr(\Hist(X_{-i}) =  t_{z} - x_{j})\\
	&=\Pr(\Hist(\vec X) =t_{z} - x_{j} + x_{k}|X_{i} = x_{k}) = \Pr(\Hist(\vec X) =t_{z+1}|X_{i} = x_{k})
\end{align*}
In other words,
\begin{align*}
&\Pr(\Hist(\vec X) \in \mt|X_{i} = x_{j}) -	\Pr(\Hist(\vec X) \in \mt|X_{i} = x_{k})\\
&= \Pr(\Hist(\vec X) = t_{q}|X_{i} = x_{j}) - \Pr(\Hist(\vec X) = t_{0}) \\
&\tab + \sum_{0\leq z<q} \Pr(\Hist(\vec X) = t_{z}|X_{i} = x_{j})	- \Pr(\Hist(\vec X) = t_{z+1}|X_{i} = x_{k})\\
&= \Pr(\Hist(\vec X) = t_{q}|X_{i} = x_{j}) - \Pr(\Hist(\vec X) = t_{0}|X_{i} = x_{k})\tag{Every term in the summation of differences cancels out.}\\
&= \Pr(\Hist(\vec X) = \ext{\mt}|X_{i} = x_{j}) - \Pr(\Hist(\vec X) =\ent{\mt}|X_{i} = x_{k})
\end{align*}
\end{proof}

\subsection{Full proof for Theorem~\ref{thm:nonunif-winner-2}: Biased Majority}\label{sec:fullproof_thm2}

\paragraph{Theorem~\ref{thm:nonunif-winner-2} (Exact DDP for Majority Rules)}
{\em Fix two candidates $\{a, b\}$ and $\Delta\subseteq \Pi(\{a, b\})$ with $|\Delta|<\infty$. For any $\alpha\in(0,1)$, the $\alpha$-biased majority rule is $(0,\delta,\Delta)$-eDDP for all $n$, where  $$\delta = \max_{p = \pi(a)\colon \pi\in\Delta}\Theta\left(\sqrt{\frac{1}{n}}\left[\left(\frac{p}{\alpha}\right)^{\alpha}\left(\frac{1-p}{1-\alpha}\right)^{1-\alpha} \right]^n\right).$$ 
In particular, $\delta = \Theta\left(\sqrt{1/n}\right)$ if $\exists \pi\in\Delta \text{ s.t. } \pi(a)=\alpha$; otherwise $\delta = \exp(-\Omega(n))$. }

\begin{proof}(Full proof for Theorem~\ref{thm:nonunif-winner-2}).

For any $\pi\in \Delta$, let $p = \pi(a)$. Let trails $\mt_a = \left\{t: t=(k, n-k) ,\, k \geq \alpha n\right\}$ and  $\mt_b= \left\{t: t=(k, n-k) ,\, k < \alpha n\right\}$. 
It follows that any histogram in $\mt_a$ results in $a$ being the winner, and any in $\mt_b$ results in $b$ as the winner. Also,  Equation (\ref{eqn:delta-hist}) implies we should {\em not} consider $\setS=\{a,b\}$ nor $\setS=\emptyset$ as otherwise $\delta = 0$ (the lower bound on $\delta$).  Thus, we only consider $\setS = \{a\}$ (when the winner is $a$, corresponding to trail $\mt_a$) or  $\setS = \{b\}$ (trail $\mt_b$). Then Equation (\ref{eqn:delta-hist}) becomes (we disregard the value of $i$ since votes are i.i.d.):
{\small
\begin{equation}\nonumber
\begin{split}
\delta &= \max_{j\in \{a, b\},x,x'} \left[\Pr\nolimits_{\vec X\sim \pi}(\Hist(\vec X)\in \mt_j | X_i = x) - \Pr\nolimits_{\vec X\sim \pi}(\Hist(\vec X)\in \mt_j | X_i = x')\right]\ \ \  \text{(Equation (\ref{eqn:delta-hist}))}\\
&= \max_{j\in\{a, b\},x,x'} \left[\Pr(\Hist(\vec X) = \ext{\mt_j} | X_i = x) - \Pr(\Hist(\vec X) = \ent{\mt_j} | X_i = x')\right]\ \ \  (\text{Lemma~\ref{lem:trail}})\\
\end{split}
\end{equation}}
We first discuss the case that $\setS =\{a\}$ where its corresponding trail $\mt_a$ starts at $\ent{\mt_a} = (n,0)$ and exits at $\ext{\mt_a} = (\lceil \alpha n \rceil, \lfloor (1-\alpha)n \rfloor)$. Here, $x = a$ and $x' = b$ maximize $\delta$. Thus, 

\begin{equation}\nonumber\Pr(\Hist(\vec X) = \ent{\mt_a} | X_i = b) = \Pr(\Hist(\vec X) = (n,0) | X_i = b) = 0
\end{equation}

and
\begin{equation}\nonumber
\begin{split}
&\Pr(\Hist(\vec X) = \ext{\mt_a} | X_i = a)\\
=\;& \Pr(\Hist(\vec X) = (\lceil \alpha n \rceil, \lfloor (1-\alpha)n \rfloor) | X_i = a)\\
=\;& \Pr(\Hist(\vec X) = (\lceil \alpha n \rceil-1, \lfloor (1-\alpha)n \rfloor))\\
=\;& p^{\lceil \alpha n \rceil-1}(1-p)^{\lfloor (1-\alpha)n \rfloor}\frac{(n-1)!}{\lceil \alpha n-1 \rceil!\cdot\lfloor (1-\alpha)n \rfloor!}\\
=\;& \Theta\left[\frac{1}{\sqrt{n}} \cdot\left(\frac{pn}{\lceil \alpha n-1 \rceil}\right)^{\lceil \alpha n-1 \rceil} \cdot\left(\frac{(1-p)n}{\lfloor (1-\alpha)n \rfloor}\right)^{\lfloor (1-\alpha)n \rfloor} \right] \;\;\;\text{(Stirling's formula)}\\
=\;& \Theta\left(\sqrt{\frac{1}{n}}\left[\left(\frac{p}{\alpha}\right)^{\alpha}\left(\frac{1-p}{1-\alpha}\right)^{1-\alpha} \right]^n\right)
\end{split}
\end{equation}

The case for $\setS =\{b\}$ is similar.  
We note that $\left(\frac{p}{\alpha}\right)^{\alpha}\left(\frac{1-p}{1-\alpha}\right)^{1-\alpha}\le 1$, and equality holds if and only if $p = \alpha$. Finally, we take the maximum of all $\delta$'s over $\pi\in\Delta$.
\end{proof}

\section{Exact Privacy of Voting Rules: General Case (Cont'd)}
In all proofs of this section, we will use $r$ instead of $\mech$ to denote GSR voting rules.

\subsection{Full proof for Theorem~\ref{thm:gsr}}\label{fullproof_theorem4}

\paragraph{Theorem~\ref{thm:gsr} (Dichotomy of  Exact DDP for GSR)}
{\em Fix $m\ge 2$ and $\Delta \subseteq \Pi({\mathcal L(C)})$ with $|\Delta|<\infty$. For any $n$, any GSR $\mech$ that satisfies monotonicity,  local stability, and canceling-out is $(0, \delta, \Delta)$-DDP, where $\delta$ is $\Theta(\sqrt{1/n})$ if $\Delta$ contains the uniform distribution over $\mathcal L(C)$, or $\exp(-\Omega(n))$ if $\Delta$ does not contain any unstable distribution. 
}

\begin{proof}[ Theorem~\ref{thm:gsr}, (Exact) DDP for GSR.] 

To present the result, we first introduce an equivalent definition of GSR that is similar to the ones used in~[Xia and Conitzer, 2009; Mossel {\em et al}., 2013].
\begin{definition}[The $(H,g_H)$ definition of GSR]
A GSR over $m$ candidates is defined by a set of hyperplanes $H=\{\vec h_1,\ldots,\vec h_{R}\}\subseteq {\mathbb R}^{m!}$ and a function $g_H:\{+,0,-\}^{|H|}\rightarrow C$. For any anonymous profile $\vec p\in {\mathbb R}^{m!}$, we let $H(\vec p)=(Sign(\vec h_1\cdot \vec p),\ldots, Sign(\vec h_{R}\cdot \vec p))$, where $Sign (x)$ is the sign ($+,-$ or $0$) of a number $x$. We let the winner be $g_{H}(H(\vec p))$.
\end{definition}
That is, to determine the winner, we first use each hyperplane in $H$ to classify the profile $\vec p$, to decide whether $\vec p$ is on the positive side ($+$), negative side ($-$), or is contained in the hyperplane ($0$). Then $g_H$ is used to choose the winner from $H(\vec p)$. We refer to this definition the $(H,g_H)$ definition. Also see Example~\ref{example:GSR} for how $(H,g_H)$ works. In the next claim, we show the equivalence of two definitions of GSR. 

\begin{myclaim}\label{claim:GSR_equivalence} The $(H,g_H)$ definition of GSR is equivalent to the $(f,g)$ definition of GSR in Definition~\ref{def:gsr}.\end{myclaim}

\begin{proof} [Proof for Claim~\ref{claim:GSR_equivalence}]
We first show that any $(H,g_H)$ GSR can be represented by a $(f,g)$ GSR in the following way: for each ranking $V$, we let $f(V)=(\vec h_1\cdot \vec e_V, h_2\cdot \vec e_V,\ldots, \vec h_R\cdot \vec e_V,0)$. Then, the $g$ function mimics $g_H$ by only focusing on orderings between the $k$th component of $f(P)$ and the last component, which is always $0$, for all $k\le R$. More precisely, ordering between the $k$th component  of $f(P)$ and $0$ uniquely determines $Sign(\vec h_k\cdot \vec p)$.

We now prove that any $(f,g)$ GSR can be represented by an $(H,g_H)$ GSR. For any pair of distinct component $k_1,k_2\le K$, we introduce a hyperplane $\vec h_{k_1,k_2}=([f(V)]_{k_1}-[f(V)]_{k_2})_{V\in L(C)}$. Therefore, for any profile $\vec p$, $\vec h_{k_1,k_2}\cdot \vec p=[f(\vec p)]_{k_1}-[f(\vec p)]_{k_2}$. The sign of $\vec h_{k_1,k_2}\cdot \vec p$ corresponds to the order between $[f(\vec p)]_{k_1}$ and $[f(\vec p)]_{k_2}$.  Then, $g_H$ mimics $g$.
\end{proof}

We are now ready to present our theorem on GSRs. We will characterize eDDP under uniform distribution and give an exponential upper bound on DDP under some other distributions.
For any pair of $\vec \pi$ and $\vec h$, we let  $\Dist(\vec \pi,\vec h)= \frac{\vec \pi \cdot \vec h}{||\vec h||_2}$ to denote the distance between hyperplane $\vec h\cdot\vec p = 0$ and vector $\vec \pi$.

	We first show that w.l.o.g.~we can assume that all hyperplanes in $H$ passes $\vec 1$.
	
	\begin{lemma} \label{lem:gsrcancel}A GSR satisfies canceling-out, if and only if there exists another equivalent GSR $r = (H, g_H)$, where all hyperplanes passes $\vec 1$.
	\end{lemma}
	\begin{proof}
		The ``if'' direction is straightforward. To prove the ``only if'' part, it suffices to prove that $g_H$ does not depend on outcomes of hyperplanes in $H$ that does not pass $\vec 1$. W.l.o.g.~let $\vec h_1\in H$ denote the hyperplane that does not pass $\vec  1$, that is, $\vec h\cdot \vec 1\ne 0$. We will prove that for any $\vec u_{-1}\in \{-1,0,1\}^{L-1}$ and any $u_1,u_1'\in \{-1,0,1\}$, such that there exist profiles $P,Q$ with $H(P)=(u_1,\vec u_{-1})$ and $H(Q)=(u_1',\vec u_{-1})$, we have $g_H(u_1,\vec u_{-1})=g_H(u_1',\vec u_{-1})$.
		
		For the sake of contradiction, suppose this does not hold and let $P,Q$ be the profiles such that $H(P)$ and $H(Q)$ differ on the first coordinate, and $r(P)\ne r(Q)$. Then, for sufficiently large $n$ we have that $H(P+n\ml(\mc))=H(Q+n\ml(\mc))$. This is because for any $\vec h\in H$ that passes $\vec 1$, we have $\vec h\cdot (P+n\ml(\mc))=\vec h\cdot P=\vec h\cdot (Q+n\ml(\mc))$. For any $\vec h\in H$ that does not pass $\vec 1$, we have $\vec h\cdot (P+n\ml(\mc))=\vec h\cdot P+n\vec h\cdot 1$, and when $n$ is sufficiently large, the sign of $\vec h\cdot (P+n\ml(\mc))$ is the same as the sign of $n\vec h\cdot 1$, which is the sign of $\vec h\cdot (Q+n\ml(\mc))$. This means that $\sign(\vec h\cdot P)=\sign(\vec h\cdot (P+n\ml(\mc)))=\sign(\vec h\cdot (Q+n\ml(\mc)))=\sign(\vec h\cdot Q)$, which is a contradiction.
	\end{proof}

	Let $r$ be a GSR, $P^*$ be the locally stable profile and $a$ be the candidate, $V,W$ be the rankings as in the statement of Definition~\ref{def:gsr-subset-properties}. W.l.o.g.~suppose $V$ is the first type ranking and $W$ is the second type ranking. In other words, $V$ (respectively, $W$) is the first (respectively, second) coordinate in the $m$-profiles space. We will show that the exact DDP bound is achieved when $S$ is the set of all profiles where the winner is $a$.
	
	We recall that for any profile $P$, a pair of different votes $V,W$. and a length $q\in\mathbb N$, $\trail{P}{V}{W}{q}$ is the trail starting at $P$, going along the $V-W$ direction, and contains $q$ profiles. We let $\trail{P}{V}{W}{\infty}=\max_q\trail{P}{V}{W}{q}$ denote the longest $V-W$ trail starting at $P$. For a GSR $r$, we define $\ed(a) = \left\{\ext{\trail{P}{V}{W}{\infty}}: \forall\,V,W\in\universe,\; r(P) = a\right\}$. In other words, there are no $W$ votes in $\ed(a)$.
	
	\begin{figure}[ht]
		\centering
		\includegraphics[width=\textwidth]{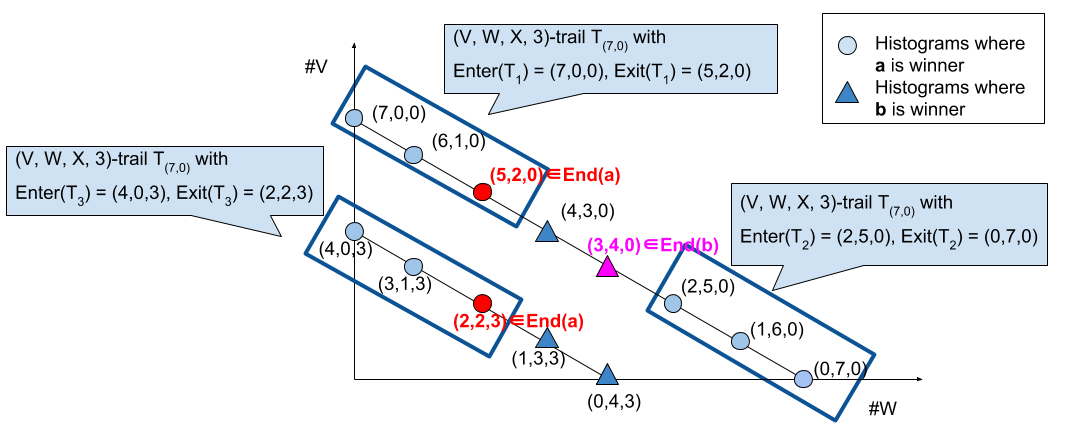}
		\caption{Example of $\ed(a)$ and $\ed(b)$, for 3-candidate case. The 3 kinds of votes other than $V,W$ and $X$ are not shown to simplify notations. Number of unshown votes are considered as constant.}
		\label{fig:end-trail}
	\end{figure}
	

	Because $r$ satisfies monotonicity, for any profile $P$ such that $r(P)=a$, we must have that $a$ is the winner under all profiles in the $V$-$W$ trail starting at $P$. Therefore, $S$ can be partitioned into multiple non-overlapping trails, each of which starts at a different profile, where $a$ is the  winner, and $a$ is no longer the winner if we go one step into the $W$-$V$ direction. Formally, we let $\ed(a)$ (shown in Figure~\ref{fig:end-trail}) denote all $n$-profiles $P$ such that (1) $r(P)=a$ and (2) $r(P+W-V)\ne a$. Then, we define a partition $\setS_a$ as follows.

	$$\setS_a = \left\{P: r(P) = a\right\} = \bigcup_{P\in \ed(a)} \trail{P}{V}{W}{\infty}$$
	It follows from Lemma~\ref{lem:trail} that
	$$\Pr(P\in \setS_a | X_1=V)-\Pr( P\in \setS_a | X_1=W)=\sum_{P\in \ed(a): P(V)>0}\Pr(P-V).$$

	We will define a subset of $n$-profile, $\mr_n$ and prove the lower bound on it. For a locally stable profile $P^*$ (with constant $\gamma$ in the statement of Definition~\ref{def:gsr-subset-properties}), let $\vec p_0=P^*-\vec 1\cdot \frac{ |P^*|}{m!}$. That is, $\vec p_0$ be obtained from $P^*$ by subtracting a constant in each component, such that $\vec p_0\cdot \vec 1=0$. For any $n$, we define $\mr_n$ to be the set of $n$-profiles that are in the $\gamma \sqrt{n}$ neighborhood of ${\frac{n}{m!}}\cdot\vec 1 +\vec p_0\cdot \sqrt{n}$ w.r.t.~$L_\infty$ norm for last $m!-2$ dimensions. That is,
	$$\mr_n=\left\{P:P[V]=0\text{ and }\forall j\ge 3, \left|P[j]-\left(\frac{n}{m!}+\vec p_0[j]\cdot \sqrt{n}\right)\right|\le \gamma \sqrt{n}\right\}$$

	Throughout the proof in Theorem~\ref{thm:gsr}, we will use $\vec \pi$ to denote the database distribution $D$, and $\pi[j]$ denote the probability of $j$-th kind of ranking. Here $P[V]$ is the number of $V$ votes in $P$ and $P[j]$ is the number of $j$-th type of vote in $P$. For any $P\in \mr_n$, we let $\piv(P)=\ed(a)\cap \trail{P}{V}{W}{\infty}$ denote the intersection of $\ed(a)$ and the $V$-$W$ trail starting at $P$. That is, $\piv(P)=P+l (V-W)$ for some $l\in \mathbb Z$, $r(\piv(P))=a$, and $r(\piv(P)-V+W)\ne a$.

	We next prove that the number of $V$ votes in $\piv(P)$ and the number of $W$ votes in $\piv(P)$ are close---the difference is $O(\sqrt{n})$.
	\begin{myclaim}\label{claim:pivdiff}
		For any $P\in \mr_n$, we have $|\piv(P)[V]-\piv(P)[W]|=O(\sqrt{n})$.
	\end{myclaim}
	\begin{proof} Let $Q^+=\piv(P)$ and $Q^-=\piv(P)-V+W$. We note that $\piv(P)$ is at the boundary of $S$, which means that $r(Q^+)\ne r(Q^-)$. Therefore, because $r$ is a GSR, the line segment between $Q^+$ and $Q^-$ must contain the intersection of $\trail{P}{V}{W}{\infty}$ and a hyperplane $\vec h\in H$. Therefore, it suffices to show that the difference in number of $V$ votes and number of $W$ votes at the intersection of $\trail{P}{V}{W}{\infty}$ and any hyperplane $\vec h$ is $O(\sqrt{n})$.
		
		We recall that by Lemma~\ref{lem:gsrcancel}, all hyperplanes for $r$ pass $\vec 1$. For any $\vec h\in H$, we recall that we assumed that $V$ and $W$ corresponds to the first and second coordinate, respectively. Because $\vec h\cdot (P+l(V-W))=0$, we have $(h_2-h_1)l=\vec h\cdot P=\vec h\cdot (P-\vec 1\cdot {\frac{n}{m!}})=O(\sqrt{n})$. This means that $|l|=|\piv(P)[V]-\piv(P)[W]|=O(\sqrt{n})$.
	\end{proof}

	\begin{myclaim} For any $P\in \mr_n$, there is a $V$-$W$ trail passing $P$.
	\end{myclaim}
	\begin{proof}
		According to the canceling out property of $r$, we can construct profile $P' = P-\frac{n-\left|P^*\right|\sqrt{n}}{m!}$, which is equivalent to $P$. For any profile $P \in \mr_n$, we have $\left|P[j]-\left(\frac{n}{m!}+\vec p_0[j]\cdot \sqrt{n}\right)\right|\le \gamma \sqrt{n}$, which is equivalent with $\left|P'[j]-P^*[j]\cdot\sqrt{n}\right|\le \gamma\sqrt{n}$, which means $\frac{P'}{\sqrt{n}}$ is in the $\gamma$ neighborhood of profile $P^*$ in terms of the $3$-rd to $m!$-th dimensions. According to the $(H,g_H)$ definition of GSR, we know $r(P^*) = r(P')$ and the claim follows by local stability of $P^*$.
	\end{proof}
	
	We will show that the probability of a subset of $\ed(a)$---the pivotal profiles on trails starting at profiles in $\mr_n$---is $\constant{1/\sqrt n}$ for the condition that $\pi$ is uniform over $\universe$.
	Let $\mr_n^-\subseteq {\mathbb R}^{m!-2}$ and for any $\vec p_-\in \mr_n^-$, we define $\piv(\vec p_-)=\piv(P)$, where $P\in \mr_n$ and $P[3,\ldots, m!]=\vec p_-$.
	%
	{\footnotesize
	\begin{align*}
	&\sum_{P\in \ed(a)}\Pr(P-V)\geq\sum_{P\in\mr_n}\Pr(\piv(P)-V)\\
	=&\sum_{\vec p_-\in \mr_n^-, |P|=n-1}\Big(\Pr(P[3,...,m!]=\vec p_-)\cdot\\
	&\;\;\;\;\;\;\;\;\;\;\;\;\;\;\;\;\;\;\;\;\;\;\;\;\Pr(P[1]=\piv(\vec p_-)[1]-1, \Pr(P[2]=\piv(\vec p_-)[2]|P[3,...,m!]=\vec p_-)\Big)\\
	=&\sum_{\vec p_-\in \mr_n^-, |P|=n-1}A(\vec p_-)B(\vec p_-)
	\end{align*}}
	where $A(\vec p_-)=\Pr(P[3,\ldots,m!]=\vec p_-)$ and $$B(\vec p_-)=\Pr(P[1]=\piv(\vec p_-)[1]-1, \Pr(P[2]=\piv(\vec p_-)[2]|P[3,\ldots,m!]=\vec p_-)$$
	It follows that $B(\vec p_-)$ is equivalent to probability of flipping a coin ($\frac{\pi[W]}{\pi[V]+\pi[W]}$ probability for head) for $\piv(\vec p_-)[1]+\piv(\vec p_-)[2]-1$ times,  with $\piv(\vec p_-)[1]-1$ heads and $\piv(\vec p_-)[2]$ tails. The next lemma gives a lower bound to $\sum_{\vec p_-\in \mr_n^-, |P|=n-1}A(\vec p_-)B(\vec p_-)$ when $\pi$ is a uniform distribution.

	\begin{lemma}\label{lemma:unifGSR} $\sum_{\vec p_-\in \mr_n^-, |P|=n-1}A(\vec p_-)B(\vec p_-) = \Omega\left(\frac{1}{\sqrt{n}}\right)$ if $\pi$ is uniform over $\universe$.
	\end{lemma}
	\begin{proof}
		We first bound the total number of $V$ and $W$ votes in $P\in\mr_n$ in the next claim.
		
		\begin{myclaim}\label{claim:pivsum} $\piv(\vec p_-)[1]+\piv(\vec p_-)[2]-1=\Theta(n)$ for all $\vec p_- \in \mr_n^-$.
		\end{myclaim}
		\begin{proof}
			$$\Big|\piv(\vec p_-)[1]+\piv(\vec p_-)[2]-\frac{2n}{m!}\Big| = \sum_{j=3}^{m!} \left|P[j]-\frac{n}{m!}\right|\leq \sum_{j=3}^{m!} \left(\gamma \sqrt{n}+|\vec{p_0}[j]| \sqrt{n}\right) \leq (\gamma+1)m!\sqrt{n}$$
		\end{proof}
		
		According to Claim~\ref{claim:pivdiff} $\&$~\ref{claim:pivsum}, we know that $B(\vec p_-)$ is equivalent to probability of flipping a fair coin for $\frac{2n}{m!}+ c_1 \sqrt{n}$ times and get $\frac{n}{m!}+ c_2 \sqrt{n}$, where $c_1$ and $c_2$ are bounded constants. In the next claim, we give a tight bound to $B(\vec p_-)$ for uniform distributed entries.
		
		\begin{myclaim}\label{claim:unifB}
			$B(\vec p_-) = \Theta\left(\sqrt{\frac{1}{n}}\right)$ for any $\vec p_- \in \mr_n^-$
		\end{myclaim}
		\begin{proof}
			Letting $n' = \frac{2n}{m!}+ c_1 \sqrt{n}$, $c' = c_2-\frac{c_1}{2}$ and assuming $n'$ is a even number, for the lower bound, we have,
			\begin{equation}
			\begin{split}
			B(\vec p_-)=& \left(\frac{1}{2}\right)^{\frac{2n}{m!}+ c_1 \sqrt{n}}\binom{\frac{2n}{m!}+ c_1 \sqrt{n}}{\frac{n}{m!}+ c_2 \sqrt{n}} = \left(\frac{1}{2}\right)^{n'}\binom{n'}{n'/2+ c' \sqrt{n}}\\
			=& \left(\frac{1}{2}\right)^{n'}\cdot {\binom{n'} {n'/2}}\cdot \frac{\frac{n'}{2}\times \cdots \times (\frac{n'}{2}-c' \sqrt{n'}+1)}{(\frac{n'}{2}+c'\sqrt{n'}-1)\times \cdots \times \frac{n'}{2}}\\
			>& \frac{1}{2^{n'}} {\binom{n'}{ n'/2}}\cdot \left(\frac{n'/2-c'\sqrt{n'}}{n'/2}\right)^{c' \sqrt{n'}}\\
			=& \Omega\left( \frac{1}{\sqrt {n}}  \right)\;\;\;\text{(applying Stirling's Formula)}
			\end{split}
			\end{equation}
			Upper bound can be obtained using similar technique as lower bound.
		\end{proof}

		The next claim gives a lower bound on $\sum_{\vec p_-\in \mr_n^-} A(\vec p_-)$. The proof 
		uses the main technique of Lindeberg-Levy Central Limit Theorem~\cite{greene2003econometric}.
	\begin{myclaim}\label{claim:unifA} $\sum_{\vec p_-\in \mr_n^-} A(\vec p_-)=\Omega  \left(1\right)$.
	\end{myclaim}
	\begin{proof}[Proof of Claim~\ref{claim:unifA}]
			We first define a set of $m!-2$ dimensions random variables that $Y_i = \left(Y_i[1],\cdots,Y_i[m!-2]\right)$, where $Y_i[j] = 1$ if ranking $j$ happens to $i$-th row and $Y_i[j] = 0$ otherwise. According to the definition of profile, we have $P[j+2] = \sum_{j=1}^n Y_i[j]$ and $\mathbb{E}(P[j]) = \frac{n}{m!}$ for uniform case. We further define a $m!-2$ dimensional random vector $\vec u$ such that $\vec{u}[j] = \left(P[j+2]-\frac{n}{m!}\right)/\sqrt{n}$, which is the scaled average of $Y_1,\cdots,Y_n$. According to Lindeberg-Levy Central Limit Theorem~\cite{greene2003econometric}, we know that the distribution of $\vec{u}$ converges in probability to multivariate normal distribution $\mathcal{N}(0,\Sigma)$, where
			$$\Sigma={
				\begin{bmatrix}
				\frac{m!-1}{(m!)^2}  &  -\frac{1}{(m!)^2}  & \cdots\ & -\frac{1}{(m!)^2}\\
				-\frac{1}{(m!)^2} &  \frac{m!-1}{(m!)^2}  & \cdots\ & -\frac{1}{(m!)^2}\\
				\vdots   & \vdots & \ddots  & \vdots  \\
				-\frac{1}{(m!)^2} & -\frac{1}{(m!)^2}  & \cdots\ & \frac{m!-1}{(m!)^2}\\
				\end{bmatrix}.
			}$$
			Since each diagonal element in $\Sigma$ is strictly larger than the sum of the absolute value of all other elements in the same row, we know that $\Sigma$ is non-singular according to Levy-Desplanques Theorem ~\cite{horn1990matrix}. According to Varah {\emph {et al.}}~\cite{varah1975lower}, we obtain a bound on $\Sigma^{-1}$'s $L_{\infty}$ norm as,
			$$||\Sigma^{-1}||_{\infty} \leq \frac{1}{\min_{i}\left(|\Sigma_{{ii}}|-\sum_{j\not=i}|\Sigma_{{ij}}|\right)}\leq \frac{(m!)^2}{2}.$$
			For any $m!-2$ dimensional random vector $\vec u$ constructed from a profile $P$ using the procedure that $\vec{u}[j] = \left(P[j+2]-\frac{n}{m!}\right)/\sqrt{n}$, we have,
			$$P\in\mr_n^-\;\;\;\;\text{if and only if}\;\;\;\;\vec{u} \in \mathbb{U}=\left\{\vec{u}:\left|\vec{u}[j]- \vec p_0[j] \right|\le \gamma,\,\forall j\in[m!-2]\right\}.$$
			Thus, for all $\vec{u}\in \mathbb{U}$ we know about its Probability Density Function (PDF) that,
			\begin{equation}\nonumber
			\begin{split}
			\pdf(\vec{u}) &= \frac{1}{\sqrt{(2\pi)^{m!-2}|\Sigma|}}\exp\left(-\frac{1}{2}\vec{u}^T\Sigma^{-1}\vec{u}\right)\\
			&= \frac{1}{\sqrt{(2\pi)^{m!-2}|\Sigma|}}  \exp\left(-\frac{1}{2}|\vec{u}^T\Sigma^{-1}\vec{u}|\right)\\
			&\geq \frac{1}{\sqrt{(2\pi)^{m!-2}|\Sigma|}}  \exp\left(-\frac{1}{2}||\vec{u}^T\Sigma^{-1}||_{\infty}\cdot||\vec{u}||_1\right)\;\text{(Holder's Inequality)}\\
			&\geq \frac{1}{\sqrt{(2\pi)^{m!-2}|\Sigma|}}  \exp\left(-\frac{1}{2}||\vec{u}^T||_{\infty}\cdot||\Sigma^{-1}||_{\infty}\cdot||\vec{u}||_1\right)\\
			&\geq \frac{1}{\sqrt{(2\pi)^{m!-2}|\Sigma|}}  \left[\exp\left(\frac{(m!)^2}{4}\right)\right]^{-||\vec{u}||_{\infty}^2}\\
			&= \Omega  \left(1\right).\\
			\end{split}
			\end{equation}
			Thus, letting $\text{Vol}(\cdot)$ be the volume function,
			$$\sum_{\vec p_-\in \mr_n^-} A(\vec p_-) \geq  \text{Vol}(\mathbb{U})\cdot\min_{\vec{u}\in\mathbb{U}}\pdf(\vec{u}) \geq \gamma^{m!-2} \cdot \Omega\left(1\right) = \Omega(1).$$
\end{proof}

		Lemma~\ref{lemma:unifGSR} follows be combining Claim~\ref{claim:unifA} and Claim~\ref{claim:unifB}.
	\end{proof}
	
	Recalling Lemma~\ref{lem:trail}, for the case that $\pi$ is uniform over all ranking, we have,
	\begin{equation}\nonumber
	\begin{split}
	\delta =& \max_{x, x', \setS} \Pr( \mech(X) \in \setS | X_1 = x ) - \Pr( \mech(X) \in \setS | X_1 = x' )\\
	\leq & \Pr( \mech(X) \in \setS_a | X_1 = W ) - \Pr( \mech(X) \in \setS_a | X_1 = V )\\
	=&  \sum_{P\in \ed(a)}\Pr(P-V)= \Omega\left(\frac{1}{\sqrt{n}}\right).
	\end{split}
	\end{equation}

	Then, we derive an upper bound of $\delta$ using the similar technique of lower bound ($\pi$ can be non-uniform for this bound). We first define $\mr_n'$, a subset of $n$-profile space, where event $P\in \mr_n'$ will be proved to happen with high probability.
	$$\mr_n'=\left\{P:P[V]=0\text{ and }\forall j\ge 3, \left|P[j]-\left(n\cdot\pi[j]\right)\right|\le  n^{3/4}\right\}.$$
	Then, we recall Lemma~\ref{lem:trail}, for the case that $\pi$ such that $\min_{i} \pi[i] > 0$, we have,
	\begin{equation}\nonumber
	\begin{split}
	\delta =& \max_{V,W, \setS} \Pr( P \in \setS | X_1 = V ) - \Pr( P \in \setS | X_1 = W )\\
	\leq & \max_{V, W} \sum_{i=1}^m \Pr( P \in \setS_i  | X_1 = V ) - \Pr( P \in \setS_i | X_1 = W ) =  \sum_{i=1}^m \sum_{P\in \ed(x_i)}\Pr(P-V).
	\end{split}
	\end{equation}
	where $\setS_i = \left\{X: r(X) = x_i\right\} = \bigcup_{P\in \ed(x_i)} \trail{P}{V}{W}{\infty}$.
	The next claim gives am upper bound on the number of pivotal profiles sharing one End.
	\begin{myclaim} For any profile $P$ in $\mr_n'$, there are at most $|H|$ pivotal profiles following $V-W$ direction.
	\end{myclaim}
	\begin{proof}
		We know from the $(H,g_H)$ definition of $GSR$ that $r$'s output only changes while passing at least one hyperplane. Considering a trail $\mt_{P_0}$ enter at $(P_0[1]+P_0[2],0,P_0[3],\cdots,P_0[m!])$ and exit at $(0,P_0[1]+P_0[2],P_0[3],\cdots,P_0[m!])$ ($P_0$ is an arbitrary $n$-profile). Thus, there are at most $|H|$ pivotal profiles sharing the same end point because $\mt_{P_0}$ passes hyperplanes at most $|H|$ times.
	\end{proof}

	Using the partition of $\mr_n'$ and arbitrarily selected candidate $a$, we have,
	{\small
	\begin{equation}\nonumber
	\begin{split}
	\sum_{P\in \ed(x_i)}\Pr(P-V) \leq & |H|\left( \sum_{P\in\mr_n'}\Pr(\piv(P)-V)+\sum_{P\in\ed(x_i)\setminus\mr_n'}\Pr(\piv(P)-V)\right)\\
	\leq & |H|\left(\sum_{\vec p_-\in \mr_n'^-, |P|=n-1}A(\vec p_-)B(\vec p_-) + \sum_{\vec p_-\not\in \mr_n'^-, |P|=n-1}A(\vec p_-)B(\vec p_-)\right)\\
	\leq & |H|\left(\max\limits_{\vec p_-\in \mr_n'^-} B(\vec p_-) \cdot \sum_{\vec p_-\in \mr_n'^-}A(\vec p_-) + \max\limits_{\vec p_-\not\in \mr_n'^-}B(\vec p_-) \cdot\sum_{\vec p_-\not\in \mr_n'^-}A(\vec p_-)\right)\\
	= & O\left(\frac{1}{\sqrt{n}}\right) \cdot O(1) + O(1) \cdot O\left(\frac{1}{\sqrt{n}}\right)\;\;\;(\text{by applying Claim~\ref{claim:GSR_piv_nonunif}})\\
	= & O\left(\frac{1}{\sqrt{n}}\right)
	\end{split}
	\end{equation}}
	
	The next claim gives an upper bound to $\sum_{\vec p_- \not\in \mr_n^-} A(\vec p_-)$.
	\begin{myclaim} $\sum_{\vec p_- \not\in \mr_n'^-} A(\vec p_-) = O\left(\frac{1}{\sqrt{n}}\right)$.
	\end{myclaim}
	\begin{proof}		
	Let $Y_j^{(i)} =$ "the $i$-th agent gives vote of type j". One can see that $P[j] = \sum_{i=1}^n Y_j^{(i)}$, $\mathbb{E}(P[j]) = n\pi[j]$ and $Var(P[j]) = n\pi[j](1-\pi[j])$.
		Thus,
		\begin{equation}\nonumber
		\begin{split}
		&\sum_{\vec p_- \not\in \mr_n^-} A(\vec p_-) = \Pr\left[\bigcup\limits_{j=3}^{m!} \left\{\ \left|P[j]-n\cdot\pi[j]\right|\le  n^{3/4}\right\} \right]\\
		\leq & \sum\limits_{j=3}^{m!} \Pr\left[ \left\{\ \Big|P[j]-\mathbb{E}(P[j])\Big|\le n^{3/4}\right\} \right]\\
		\leq & \sum\limits_{j=3}^{m!} \frac{n\pi[j](1-\pi[j])}{ n^{3/2}} \;\;\;(\text{by Chebyshev's Inequality})\\
		= & O\left(\frac{1}{\sqrt{n}}\right)\\
		\end{split}
		\end{equation}
	\end{proof}
	
	Then, all we need is an upper bound on $B(\vec p_-)$, and we first prove that the length of $V-W$ sequence is $\Theta(n)$ for all $P\in \mr_n'$.
	\begin{myclaim}\label{claim:GSR_piv_nonunif}
		$\piv(\vec p_-)[1]+\piv(\vec p_-)[2]-1=\Theta(n)$ for all $P\in \mr_n'$.
	\end{myclaim}
	\begin{proof}
		$$\left|\piv(\vec p_-)[1]+\piv(\vec p_-)[2]-n(\pi[W]+\pi[V])\right|= \sum_{j=3}^{m!} \left|P[j]-n\cdot\pi[j]\right| \leq \sum_{j=3}^{m!}  n^{3/4}\leq m!\cdot n^{3/4}$$
	\end{proof}
	
	Then, using the same technique of Claim~\ref{claim:unifB}, we know that,
	$$B(\vec p_-) = \Theta\left(\sqrt{\frac{1}{n}}\right)\;\;\; \text{for all} \;\;\;p_- \in \mr_n'^-$$
	Thus, combining all results above, we have,
	\begin{equation}\nonumber
	\begin{split}
	\delta \leq &  \sum_{i=1}^m \sum_{P\in \ed(x_i)}\Pr(P-V) = \sum_{i=1}^m \sum_{P\in \ed(x_i)}\Pr(P-V) = O\left(\frac{1}{\sqrt{n}}\right)\\
	\end{split}
	\end{equation}


	Next, we will give a exponential (tighter) upper bound on $\delta$ when $\pi$ does not belong to any hyperplanes.We first give a generalized definition of pivotal profile.
	\begin{definition}[Generalized Pivotal Profile]\label{def:gsr-nontrivalEE}
		Profile $P$ is a (generalized) pivotal profile if there exist pair of votes $V$ and $W$ such that $r(P)\neq r(P-V+W)$.
	\end{definition}
	Then, we define a distance function $\gdist(P, h)$ to be a generalized distance between profile $P$ and hyperplane $h$. We define 
	$$\gdist(P,\vec h) = \inf_{P'\in \Tilde{h}} ||P-P'||_2,$$
	where $\Tilde h = \left\{P\in h: \exists \text{ unit vector } \vec e \text{ s.t. } r(P'-\vec e)\neq r(P'+\vec e)  \right\}$. In the next lemma we will show generalized pivotal profiles only lays close to hyperplanes. We fist gives definition of distance function $\Dist(\cdot,\cdot)$:\\
	1. for hyperplane $h$ and a point ($n$-profile) $P$, $\Dist(h,P) = \frac{\vec h \cdot P}{||\vec h||_2}$, which is the Euclidean distance between $P$ and hyperplane $\vec h \cdot \vec p = 0$.\\
	2. for 2 points ($n$-profile) $P_1$ and $P_2$, $\Dist(h,P)$, returns the Euclidean distance between $P_1$ and $P_2$.
	\begin{myclaim}\label{claim:newdist}
		For any GSR $r = (H,g_H)$ and one of its generalized pivotal profile $P$, there must exist one hyperplane $\vec{h}\in H$ such that $\Dist(h, P) \leq \sqrt{2}$.
	\end{myclaim}
	\begin{proof}
		Recalling the definition of generalized pivotal profiles, we know the GSR winner will change at the $1$ neighborhood of $P$. Thus, there must exist a hyperplane $\vec{h}\in H$ and pair of votes $V,W$ such that $\sign\left[\vec{h}\cdot P\right] \not=  \sign\left[\vec{h}\cdot (P+V-W)\right]$ and $\Dist(h,\, P) \leq \Dist(P, P+V-W) = \sqrt{2}$.
	\end{proof}

	\begin{lemma}\label{lem:exp_bound}
		Let $D$ be the distribution on profiles (databases of votes), where each entry is iid according to distribution $\pi$ over linear orders on $m$ candidates. GSR $r(H,h_H)$ is $(0,\delta, \Delta = \{(D, \emptyset)\})$-DDP when only the winner is announced, where
		$$\delta = O\left[\exp\left(-\frac{\left[\min_{h\in H}\gdist(\vec {\pi},h)\right]^2}{3(m!)^2\left(\max_{i\in [m!]}\pi[i]\right)}\cdot n\right)\right] = O\left[e^{-\Omega(n)}\right].$$
	\end{lemma}
	\begin{proof}
		We first define the set of all generalized pivotal profiles $\mathbb{P}_{\text{Piv}}$. For any $P\in\mathbb{P}_{\text{Piv}}$, we know that there exist hyperplane $h\in H$ such that $\gdist(h,\, P) \leq \sqrt{2}$. According to triangular inequality, we have $\gdist(n\vec {\pi},\, P) \geq \gdist(n\vec {\pi},\, h)-\Dist(h,\, P) \geq n\gdist(\vec {\pi},\, h) - \sqrt{2}$. The second $\geq$ sign comes from  the fact that all hyperplanes passes $\vec 0$. Thus, there must exist one dimension $j$ that $|P[j]-n\pi[j]| \geq \frac{n\gdist(\vec {\pi},\, h)-\sqrt{2}}{m!}$. Then, we bound $\delta$ as,
		\begin{equation}\nonumber
		\begin{split}
		\delta =& \max_{V,W, \setS} \left[\Pr( P \in \setS_i | X_1 = V ) - \Pr( P \in \setS_i | X_1 = W )\right]\\
		\leq& \sum_{P \in \mathbb{P}_{\text{Piv}}}\left[ \max_{V} \Pr( P \in \mathbb{P}_{\text{Piv}} | X_1 = V )\right]\\
		\leq& \max_{V,h,j} \Pr\left( |P[j]-n\pi[j]| \geq  \frac{n\gdist(\vec {\pi},\, h)-\sqrt{2}}{m!} \bigg| X_1 = V \right)\\
		\leq& \max_{h,j} \Pr\left( |P[j]-n\pi[j]| \geq  \frac{n\gdist(\vec {\pi},\, h)-\sqrt{2}}{m!}-1 \right)\\
		=& O\left[\exp\left(-\frac{\left[\min_{h\in H}\gdist(\vec {\pi},h)\right]^2}{3(m!)^2\left(\max_{i\in [m!]}\pi[i]\right)}\cdot n\right)\right] \text{by applying Chernoff bound.}\\
		\end{split}
		\end{equation}
	\end{proof}

Theorem~\ref{thm:gsr} follows by combining all three bounds derived above.
	
\end{proof}

\subsection{Proof for Corollary~\ref{cor:gsr-subset}}

\begin{proposition}\label{prop:gsr-subset-positional} All positional scoring rules and all Condorcet consistent and monotonic rules satisfy all properties required by Theorem~\ref{thm:gsr}.
\end{proposition}

\begin{proof} [Proof of Proposition~\ref{prop:gsr-subset-positional} 
]
Suppose $s_1=\cdots=s_l>s_{l+1}$. We let $V=[a\succ c_1\succ c_{l-1}\succ b\succ \text{others}]$ and $W=[\succ c_1\succ c_{l-1}\succ b\succ a\succ \text{others}]$. Let $M$ be the permutation $c_1\rightarrow c_2\rightarrow\ldots c_{m-2}\rightarrow c_1$. Let $V_1=[a\succ b\succ \text{others}]$  and $V_2=[b\succ a\succ \text{others}]$. Let $P=\bigcup_{i=1}^{m-2}M^i(V_1)\cup M^i(V_2)$. Let  $P^*=2P'\cup\{V,W\}$. It follows that $a$ and $b$ are the only two candidates tied in the first place in $P^*$. Therefore, there exists $\epsilon$ to satisfy the condition.

The same profile can be used to prove the local stability of all Condorcet consistent and monotonic rules.
\end{proof}
Corollary~\ref{cor:gsr-subset} follows by the definition of voting rules and the definition of positional scoring rules.

\subsection{Exact DDP for Histogram}
As a complementary result to the DDP result for histograms, we present the histogram's eDDP with $\epsilon = 0$.

\begin{theorem}[Exact DDP of Histogram]\label{thm:nonunif-histogram-c}
Fix $\bins\ge 2$, $\universe = \{x_1, \cdots, x_\bins\}$, and $\Delta \subseteq \Pi(\universe)$. Let $p_{\min} = \min_{\pi\in\Delta} (\pi(x_i) +\pi(x_j))$. For all $n \in \Natural$, $\Hist$ of $n$ voters is $(0, \delta(n) = \Theta\left(\sqrt{\frac{1}{n p_{\min}}}\right), \Delta)$-\eDDP.
\end{theorem}
\begin{proof}[Sketch]
	First we present the case for $\bins = 2$.
	\begin{lemma}[Exact DDP for Histogram, when $\bins = 2$]\label{nonunif-hist-2}
		Fix $\universe = \{x_{1}, x_{2}\}$ and $\Delta \subseteq \Pi(\universe)$. The histogram for $n$ voters is $(0, \constant{1/\sqrt{n}}, \Delta)$-\eDDP.
	\end{lemma}
	
	\begin{proof}[Lemma~\ref{nonunif-hist-2}]
		Consider some $\pi\in \Delta$, and let $p = \pi(a)$. Without loss of generality let $x = x_{1}$ and $x' = x_{2}$ (otherwise, rename them). Then, the maximizing set $\setS$ in Equation (\ref{equ:delta}) is exactly the set of histograms such that $$\Pr_{\vec X \in \pi}(\Hist(\vec X)\in \setS | X_{i} = x_{1}) > \Pr(\Hist(\vec X)\in \setS | X_{i} = x_{2})$$ Since votes are i.i.d., these follow the binomial distribution (with $n$ trials). Below we find that $\setS$ is the set of histograms $(k, n-k)$ where $k > pn$.
		\begin{align*}
		\Pr(\Hist(\vec X) = (k, n-k)|X_{i} = x_{1}) &> \Pr(\Hist(\vec X) = (k, n-k)|X_{i} = x_{2}) \\
		\implies p^{k-1}(1-p)^{n-k}\frac{(n-1)!}{(n-k)!(k-1)!} &> p^{k}(1-p)^{n-k-1} \frac{(n-1)!}{(n-k-1)!k!}\\
		\implies k &> pn
		\end{align*}
		
		Thus, $\setS = \left\{t = (k, n-k) \colon k > pn \right\}$. This set forms a trail $\mt$ which starts from $\ent(\mt) = (n, 0)$ and exits at $\ext{\mt} = (pn+1, n-(pn+1))$. Thus,
		\begin{align*}
		\delta &=  \Pr(\Hist(\vec X)\in \setS | X_{i} = x_1) - \Pr(\Hist(\vec X)\in \setS | X_{i} = x_2) \tag{Equation (\ref{equ:delta})}\\
		&=\Pr(\Hist(\vec X) = \ext{\mt} |X_{i} = x_{1}) -  \Pr(\Hist(\vec X) = \ent{\mt} | X_{i} = x_{2})\tag{Lemma~\ref{lem:trail}}\\
		&= \Pr(\Hist(\vec X) = (pn+1, n-(pn+1)) |X_{i} = x_{1}) - \Pr(\Hist(\vec X) = (n, 0) | X_{i} = x_{2})\\
		&= p^{pn}(1-p)^{n-pn-1}\frac{(n-1)!}{(pn)!(n-pn-1)!}\tag{When one row is fixed to $x_{2}$, the probability of histogram being $(n, 0)$ is zero.}\\
		&= \constant{1/\sqrt{n}} \tag{By applying Stirling's formula}
		\end{align*}
	\end{proof}
	
	We can generalize the result to $\bins >2$, by using the trail technique. Again we assume WLOG that $x = x_{1}$ and $x' = x_{2}$. Let $t = (t_1, \cdots, t_\bins)$ be the histogram, where $t_i$ counts the number of occurrences of $x_i$. We observe that, when votes are  i.i.d,  $t_{3}, \cdots, t_{\bins}$ are independent of $t_{1}, t_{2}$ when conditioned on the sum $s = t_{1} + t_{2}$. This means that we can compute $\delta$ for general $\bins$, as a sum  $$\delta = \sum_{0<s\leq n} \delta_{s} \Pr(\funvar{Bin}(n, \pi(x_1) + \pi(x_2)) = s)$$ Where $\delta_{s}$ is the $\delta$-value for $\bins = 2$, when there are $s$ votes. Using Chernoff bound we see that $\funvar{Bin}(n,\pi(x_1) + \pi(x_2))$ is concentrated at its mean $n(\pi(x_1) + \pi(x_2))$. Plugging in the result for $\bins = 2$, we get $\delta = \Theta\left(\frac{1}{\sqrt{n (\pi(x_1) + \pi(x_2))}}\right) $.
\end{proof}

\subsubsection{Full proof} Below we present the full proof of Theorem~\ref{thm:nonunif-histogram-c}, using Lemma~\ref{nonunif-hist-2} which showed the case for $\bins = 2$.\\

\begin{proof}[Proof of Theorem~\ref{thm:nonunif-histogram-c}, Exact DDP of Histogram]	
	
	Consider any $\pi\in \Delta$, and let $p_i = \pi(x_i)$.
	Like in the $\bins = 2$ case, without loss of generality, we can let $x = x_{1}$ and $x' = x_{2}$ (otherwise, rename them). Then, the maximizing set $\setS$ (similar to when $\bins = 2$) is exactly the set of histograms such that $$\Pr_{\vec X\sim \pi}(\Hist(\vec X)\in \setS | X_{i} = x_{1}) > \Pr_{\vec X\sim \pi}(\Hist(\vec X)\in \setS | X_{i} = x_{2})$$
	(We will implicitly assume ${\vec X\sim \pi}$ from now on) Since we have i.i.d. votes, the histogram follows the multinomial distribution (with $n$ trials). For any $0< s \leq n$, $(t_{3}, \cdots, t_{\bins})$ where $t_3 +\cdots + t_{\bins} = n - s$, and $k\leq s$:
	\begin{align*}
	\Pr(\Hist(\vec X) = (k, s-k, t_{3}, \cdots, t_{\bins})|X_{i} = x_{1}) &> \Pr(\Hist(\vec X) = (k, s-k, t_{3}, \cdots, t_{\bins})|X_{i} = x_{2}) \\
	p_{1}^{k-1}p_{2}^{n-k}p_{3}^{t_{3}}\cdots p_{\bins}^{t_{\bins}}\frac{(n-1)!}{(k - 1)!(s-k)!t_{3}!\cdots t_{\bins}!} &> p_{1}^{k}p_{2}^{n-k-1}p_{3}^{t_{3}}\cdots p_{\bins}^{t_{\bins}} \frac{(n-1)!}{(s-k-1)!k!t_{3}!\cdots t_{\bins}!}\\
	\frac{p_{2}}{s-k} &> \frac{p_{1}}{k} \\
	k &> \left(\frac{p_{1}}{p_{1} + p_{2}}\right)s
	\end{align*}
	
	Thus, the set $\setS = \left\{t= (k, s-k, t_{3}, \cdots, t_{\bins})\colon k> \left(\frac{p_{1}}{p_{1} + p_{2}}\right)s\right\}$.
	
	Let $p = \frac{p_{1}}{p_{1} + p_{2}}$. For each $0< s \leq n$ and $(t_{3}, \cdots, t_{\bins})$ which sum to $n - s$ (i.e. $t_3+\cdots+t_{\bins} = n-s$), let $\mt_{s, (t_{3}, \cdots, t_{\bins})}$ be the trail starting from $\ent{\mt_{s, (t_{3}, \cdots, t_{\bins})}} = (s, 0, t_{3}, \cdots, t_{\bins})$ and exiting at $\ext{\mt_{s, (t_{3}, \cdots, t_{\bins})}} = (ps + 1, s - (ps + 1), t_{3}, \cdots, t_{\bins})$. The set $\setS$ then can be partitioned into such trails. Thus,
	\begin{align*}
	\delta &= \Pr(\Hist(\vec X)\in \setS | X_{i} = x_{1}) - \Pr(\Hist(\vec X)\in \setS | X_{i} = x_{2})\\
	&=\sum_{\mt_{s, (t_{3}, \cdots, t_{\bins})}} \Pr(\Hist(\vec X)\in \mt_{s, (t_{3}, \cdots, t_{\bins})} | X_{i} = x_{1}) - \Pr(\Hist(\vec X)\in \mt_{s, (t_{3}, \cdots, t_{\bins})} | X_{i} = x_{2})\\
	&= \sum_{\mt_{s, (t_{3}, \cdots, t_{\bins})}} \Pr(\Hist(\vec X) = \ext{\mt_{s, (t_{3}, \cdots, t_{\bins})}} |X_{i} = x_{1}) \\
	&\tab\tab -  \Pr(\Hist(\vec X) = \ent{\mt_{s, (t_{3}, \cdots, t_{\bins})}} | X_{i} = x_{2})\tag{By Lemma~\ref{lem:trail}}\\
	&= \sum_{0 < s \leq n} \sum_{\substack{(t_{3}, \cdots, t_{\bins})\\t_{3}+\cdots+t_{\bins} = n-s}} \Pr(\Hist(\vec X) =(ps + 1, s - (ps + 1), t_{3}, \cdots, t_{\bins}) |X_{i} = x_{1})\\
	&\tab \tab \tab \tab -\Pr(\Hist(\vec X) = (s, 0, t_{3}, \cdots, t_{\bins}) |X_{i} = x_{2}) 
	\end{align*}
	
	Now let us consider  these two probabilities. Consider the distribution $X_{-i}$, which is $\vec X$ but without the $i$th row. Let the random variables of the individual components of $\Hist(X_{-1})$ be $(a_{1}, \cdots, a_{\bins})$. Since votes are i.i.d., for any $(t_{1}, \cdots, t_{\bins})$,
	\begin{align*}
	&\Pr(\Hist(\vec X) = (t_{1}, \cdots, t_{\bins}) |X_{i} = x_{1}) \\
	&= \Pr(\Hist(X_{-i}) = (t_{1} - 1, t_{2}, t_{3}, \cdots, t_{\bins}))\\
	&=   \Pr((a_{1}, \cdots, a_{\bins}) = (t_{1} -1, t_{2}, t_{3}, \cdots, t_{\bins})) \tag{Recall these $a$'s are components of $\Hist(X_{-i})$}\\
	&= \Pr((a_{1}, \cdots, a_{\bins}) = (t_{1} -1, t_{2}, t_{3}, \cdots, t_{\bins}) | a_{1} + a_{2} = s) \times \Pr( a_{1}+ a_{2} = s) \\
	&= \Pr (\left(a_{1}, a_{2}\right) = \left(t_{1}-1, t_{2}\right) | a_{1} + a_{2} = s)\\
	&\tab \times \Pr( \left(a_{3}, \cdots, a_{\bins}\right)=\left(t_{3}, \cdots, t_{\bins}\right) | a_{1} + a_{2} = s)\times \Pr (a_{1} + a_{2} = s)\tag{By Lemma~\ref{lem:cond-indep}, $(a_{1}, a_{2})$ and $(a_{3}, \cdots, a_{\bins})$ are independent conditioned on $a_{1} + a_{2} = s$}
	\end{align*}
	
	Similar to the $\bins = 2$ case, $\Pr(\Hist(\vec X) = (s, 0, t_{3}, \cdots, t_{\bins}) |X_{i} = x_{2}) = 0$. This is because when one vote is fixed to $x_{2}$, it is impossible to have zero in the second component in the histogram (which is the number of occurences of $x_{2}$). Thus,
	\begin{align*}
	\delta &= \sum_{0 < s \leq n} \sum_{\substack{(t_{3}, \cdots, t_{\bins})\\t_{3}+\cdots+t_{\bins} = n-s}} \Pr((a_{1}, a_{2}) = (ps, s - (ps + 1))| a_{1} + a_{2} = s)\\
	&\tab \tab \tab \tab \times \Pr( (a_{3}, \cdots, a_{\bins})=(t_{3}, \cdots, t_{\bins}) | a_{1} + a_{2} = s)\times \Pr (a_{1} + a_{2} = s)\\
	&= \sum_{0 < s \leq n} \Pr((a_{1}, a_{2}) = (ps, s - (ps + 1))| a_{1} + a_{2} = s) \times \Pr (a_{1} + a_{2} = s)\\
	&\tab \tab \tab \tab \times \sum_{\substack{(t_{3}, \cdots, t_{\bins})\\t_{3}+\cdots+t_{\bins} = n-s}}  \Pr( (a_{3}, \cdots, a_{\bins})=(t_{3}, \cdots, t_{\bins}) | a_{1} + a_{2} = s) \tag{Factor out the common terms $ \Pr((a_{1}, a_{2}) = (ps, s - (ps + 1))| a_{1} + a_{2} = s)$ and $\Pr (a_{1} + a_{2} = s)$} \\
	&=  \sum_{0 < s \leq n} \Pr((a_{1}, a_{2}) = (ps, s - (ps + 1))| a_{1} + a_{2} = s) \times \Pr (a_{1} + a_{2} = s)\tag{For any $s$, the second sum equals one.}
	\end{align*}
	
	Where $\Pr((a_{1}, a_{2}) = (ps, s - (ps + 1))| a_{1} + a_{2} = s)$ is the $\delta$ value for histogram when $\bins = 2$, the vote distribution is $\pi'\in \Pi(\{x_1, x_2\})$, where $\pi'(x_1) = \frac{p_{1}}{p_{1} + p_{2}}$, and number of voters is $s$ (we refer to Lemma~\ref{nonunif-hist-2} of the $\bins = 2$ case for this claim). We denote this $\delta$ by $\delta_{s}$. Moreover,
	
	\begin{align*}
	\Pr (a_{1} + a_{2} = s) &= \Pr(\funvar{Bin}(n, p_1 + p_2) = s)
	\end{align*}
	We denote $p' = p_1 + p_2$. Then,  $\funvar{Bin}(n, p')$ is the binomial distribution with $n$ trials and probability $p' = p_1 + p_2$ (recall that $p_i = \pi(x_i)$). Then
	
	\begin{align*}
	\delta &= \sum_{0 < s \leq n}\delta_{s} \Pr\left(\funvar{Bin}\left(n, p'\right) = s\right)\\
	&= \sum_{\substack{s \geq \left(1 - \sqrt{\frac{3}{4}}\right)np' \\ s\leq \left(1 + \sqrt{\frac{3}{4}}\right)np'}} \Pr\left(\funvar{Bin}\left(n, p'\right) = s\right) \times  \delta_{s} + \sum_{\substack{s < \left(1 - \sqrt{\frac{3}{4}}\right) np' \\ s> \left(1 + \sqrt{\frac{3}{4}}\right) np'}} \Pr\left(\funvar{Bin}\left(n, p'\right) = s\right) \times  \delta_{s}
	\end{align*}

	Lower bound of $\delta$:
	{\small
	\begin{align*}
	\delta &\geq \sum_{\substack{s \geq \left(1 - \sqrt{\frac{3}{4}}\right) np' \\ s\leq \left(1 + \sqrt{\frac{3}{4}}\right) np'}} \Pr\left(\funvar{Bin}\left(n, p'\right) = s\right) \times  \delta_{s}\\
	&\text{Since $\delta_{s}$ decreases with larger $s$ (more votes implies more privacy), $\delta_{\left(1+ \sqrt{\frac{3}{4}}\right)np'}$ is the minimum.}\\
	&\geq \delta_{\left(1+ \sqrt{\frac{3}{4}}\right)np'} \ \times \sum_{\substack{s \geq \left(1 - \sqrt{\frac{3}{4}}\right)np' \\ s\leq \left(1 + \sqrt{\frac{3}{4}}\right) np'}} \Pr\left(\funvar{Bin}\left(n, p'\right) = s\right)  \\
	&= \delta_{\left(1+ \sqrt{\frac{3}{4}}\right)np'}\ \times \left[1 - 
	\Pr\left(\funvar{Bin}(n, p') > \left(1+ \sqrt{\frac{3}{4}} np'\right)\right)- \Pr\left(\funvar{Bin}(n, p') < \left(1 - \sqrt{\frac{3}{4}} np'\right)\right)\right] 
	\end{align*}}
	By Chernoff bound for binomial distribution, for any $0 < \beta < 1$, we have:
	\begin{align*}
	\Pr \left( \funvar{Bin} \left( n, p' \right) > \left(1 + \beta\right) \mu\right)  &\leq e^{-\frac{\beta^{2} \mu}{3}}\\
	\Pr\left(\funvar{Bin}\left(n, p'\right) < \left(1 - \beta\right) \mu\right) &\leq e^{-\frac{\beta^{2} \mu}{2}} 
	\end{align*} 
	Where $\mu = np'$ is the mean of $\funvar{Bin}\left(n, np'\right)$. Now let $\beta = \sqrt{\frac{3}{4}}$, which is between 0 and 1. Then, 
	\begin{align*}
	1 &\geq 	\left[1 - 
	\Pr\left(\funvar{Bin}\left(n, p'\right) > \left(1+ \sqrt{\frac{3}{4}}\right) np'\right)- \Pr\left(\funvar{Bin}\left(n, p'\right) < \left(1 - \sqrt{\frac{3}{4}} \right) np'\right)\right] \\
	&\geq 1 - e^{-\frac{3}{4}\frac{\mu}{3}} - e^{-\frac{3}{4}\frac{\mu}{2}} \\
	&= 1 - e^{-\frac{np'}{2}} - e^{-\frac{3np'}{2}} \\
	&\tag{For large enough $n$,  $np' \geq 1$, so $e^{-\frac{np'}{2}}\leq e^{-1/2}$ and $e^{-\frac{3np'}{2}}\leq e^{-3/2}$ }\\
	&\geq 1 - e^{-1/2} - e^{-3/2} \geq \frac{1}{10}
	\end{align*}
	Which means $\Bigg[1 - 
	\Pr\Big(\funvar{Bin}(n, p') > \left(1+ \sqrt{\frac{3}{4}}\right) np'\Big)- \Pr\Big(\funvar{Bin}(n, p') < \left(1 - \sqrt{\frac{3}{4}}\right) np'\Big)\Bigg] = \Theta(1)$.\\
	By Stirling formula, we have 
	\begin{align*}
	\delta_{\left(1+ \sqrt{\frac{3}{4}}\right)np'} &= \Theta\left(\frac{1}{\sqrt{\left(1+ \sqrt{\frac{3}{4}}\right)np'}}\right)\\
	&= \Theta\left({\sqrt{\frac{1}{np'}}}\right) \tag{Recall we assumed the maximizing $x, x'$ are $x_{1}, x_{2}$, up to renaming the $x_{i}$'s, and that $p' = p_{1} + p_{2}$}\\
	& = \Theta\left({\sqrt{\frac{1}{np_{\min}}}}\right) \tag{In general, $p_{\min} = \min_{i \not= j\in[\bins]}(p_{i} + p_{j})$.}
	\end{align*}
	Which gives us the lower bound $\delta \geq \Theta\left({\sqrt{\frac{1}{np_{\min}}}}\right)$.
	
	Upper bound of $\delta$:
	{\small
	\begin{align*}
	\delta &= \sum_{\substack{s \geq \left(1 - \sqrt{\frac{3}{4}}\right) np' \\ s\leq \left(1 + \sqrt{\frac{3}{4}}\right) np'}} \Pr\left(\funvar{Bin}\left(n, p'\right) = s\right) \times  \delta_{s}\\
	&\tab + \sum_{\substack{s < \big(1 - \sqrt{\frac{3}{4}}\big) np' \\ s> \big(1 + \sqrt{\frac{3}{4}}\big) np'}} \Pr\left(\funvar{Bin}\left(n, p'\right) = s\right) \times  \delta_{s}\\
	&\text{Since $\delta_{s} \leq 1$ for all $s$ and $\sum_{\substack{s \geq \big(1 - \sqrt{\frac{3}{4}}\big) np' \\ s\leq \big(1 + \sqrt{\frac{3}{4}}\big) np'}}\Pr\left(\funvar{Bin}\left(n, p'\right) = s\right)  \leq 1$}\\
	&\leq \max_{\big(1 - \sqrt{\frac{3}{4}}\big) np'\ \leq s\ \leq \big(1+ \sqrt{\frac{3}{4}}\big)np'}(\delta_{s})\ \ + \sum_{\substack{s < \big(1 - \sqrt{\frac{3}{4}}\big) np' \\ s> \big(1 + \sqrt{\frac{3}{4}}\big) np'}}\Pr\left(\funvar{Bin}\left(n, p'\right) = s\right)  \\
	&= \delta_{\big(1 - \sqrt{\frac{3}{4}}\big) np'}\ \ + \  \Pr\left(\funvar{Bin}\left(n, p'\right) < \left(1 - \sqrt{\frac{3}{4}}\right) np'\right) +  \Pr\left(\funvar{Bin}\left(n, p'\right) > \left(1+ \sqrt{\frac{3}{4}}\right) np'\right)\\
	&\leq  \delta_{\big(1 - \sqrt{-\frac{3}{4}}\big) np'}\ + \ e^{-\frac{np'}{2}} + e^{\frac{3np'}{2}} \tag{By Chernoff bound for binomial}\\
	&\leq \delta_{\big(1 - \sqrt{-\frac{3}{4}}\big) np'}\ + \ 2\sqrt{\frac{1}{np'}} \tag{Since $np' \geq 0$, both $e^{-\frac{np'}{2}}, e^{\frac{3np'}{2}} \leq \sqrt{\frac{1}{np'}}$}\\
	&\text{By Stirling's formula, }	\delta_{\big(1 - \sqrt{-\frac{3}{4}}\big) np'}\ = \const\left(\frac{1}{\sqrt{\big(1 - \sqrt{-\frac{3}{4}}\big) np'}}\right)\\
	&= \Theta\left(\sqrt{\frac{1}{np'}}\right)
	\end{align*}}
	As is with the lower bound, in general (without assuming $(x, x') = (x_{1}, x_{2})$), we have $p' = p_{\min} = \min_{i \not=j\in [\bins]} (p_{i} + p_{j})$.
	Since both lower and upper bounds of $\delta$ are $\Theta\left(\sqrt{\frac{1}{np_{\min}}}\right)$, $\delta = \Theta\left(\sqrt{\frac{1}{np_{\min}}}\right)$.
\end{proof}
\begin{lemma}[Conditional independence]\label{lem:cond-indep}
	Let $\universe = \{x_{1}, \cdots, x_{\bins}\}$ and $\pi \in \Delta(\universe)$. 
	Let $\#x_{i}$ denote the r.v. of the number of occurrences of the vote $x_{i}$ in $\pi$. Then, for all $0\leq s\leq n$,
	the random variables $(\#x_{1}, \#x_{2})$ and $(\#x_{3},\cdots, \#x_{\bins})$ are independent conditioned on $\#x_{1} + \#x_{2} = s$. In other words, for any $(t_{1}, \cdots, t_{\bins})$ such that $\sum_{i} t_{i} = n$, we have
	{\small
	\begin{align*}
	&\Pr((\#x_{1},\cdots, \#x_{\bins}) = (t_{1}, \cdots, t_{\bins}) \ |\ \#x_{1} + \#x_{2} = s)\\
	=& \Pr((\#x_{1}, \#x_{2}) = (t_{1}, t_{2})\ |\ \#x_{1} + \#x_{2} = s) \times \Pr((\#x_{3},\cdots, \#x_{\bins}) = (t_{3}, \cdots, t_{\bins})\ |\ \#x_{1} + \#x_{2} = s)
	\end{align*}}
\end{lemma}

\begin{proof}[Proof for Lemma~\ref{lem:cond-indep}]
	We equivalently show that
	\begin{equation}\label{eqn:cond-indep}
	\begin{split}
	&\Pr((\#x_{3},\cdots, \#x_{\bins}) = (t_{3}, \cdots, t_{\bins})\ |\ \#x_{1} + \#x_{2} = s) \\
	&= \Pr((\#x_{3},\cdots, \#x_{\bins}) = (t_{3}, \cdots, t_{\bins}) \ |\ \#x_{1} + \#x_{2} = s\land (\#x_{1}, \#x_{2} )= (t_{1}, t_{2}))
	\end{split}
	\end{equation}
	Now, conditioned on there being exactly $s$ people who voted $x_1$ or $x_2$, let $D_{1}> D_{2}>  \cdots >D_{s}$ denote the random variables of the indices of the votes in the profile which voted for $x_{1}$ or $x_{2}$, in ascending order. By total probability, the left hand side of Equation~\ref{eqn:cond-indep} is:
	{\small
	\begin{align*}
	\Pr&((\#x_{3},\cdots, \#x_{\bins}) = (t_{3}, \cdots, t_{\bins})\ |\ \#x_{1} + \#x_{2} = s)\\
	=& \sum_{d_{1}> d_{2} > \cdots > d_{s}} \Pr((\#x_{3},\cdots, \#x_{\bins}) = (t_{3}, \cdots, t_{\bins})\ |\ \#x_{1} + \#x_{2} = s \land (D_{1}, \cdots, D_{s})  = (d_{1}, \cdots, d_{s}))\\
	&\tab \tab \tab \times \Pr((D_{1}, \cdots, D_{s})  = (d_{s}, \cdots, d_{s})\ |\  \#x_{1} + \#x_{2} = s)
	\end{align*}}
	We already assume there are exactly $s$ votes for $x_{1}$ or $x_{2}$, so
	\begin{align*}
	&\Pr((\#x_{3},\cdots, \#x_{\bins}) = (t_{3}, \cdots, t_{\bins})\ |\ \#x_{1} + \#x_{2} = s)\\
	&= \sum_{d_{1}> d_{2} > \cdots > d_{s}} \Pr((\#x_{3},\cdots, \#x_{\bins}) = (t_{3}, \cdots, t_{\bins})\ |\  (D_{1}, \cdots, D_{s})  = (d_{1}, \cdots, d_{s}))\\
	&\tab \tab \tab \times \Pr((D_{1}, \cdots, D_{s})  = (d_{1}, \cdots, d_{s}))\\
	\end{align*}
	The right hand side of Equation~\ref{eqn:cond-indep} is:
	{\footnotesize\begin{align*}
	&\Pr((\#x_{3},\cdots, \#x_{\bins}) = (t_{3}, \cdots, t_{\bins}) \ |\ \#x_{1} + \#x_{2} = s\land (\#x_{1}, \#x_{2} )= (t_{1}, t_{2}))\\
	&= \Pr((\#x_{3},\cdots, \#x_{\bins}) = (t_{3}, \cdots, t_{\bins}) \ |\ \#x_{1} + \#x_{2} = s\land (\#x_{1}, \#x_{2} )= (t_{1}, t_{2}))\\
	&= \Pr((\#x_{3},\cdots, \#x_{\bins}) = (t_{3}, \cdots, t_{\bins}) \ |\ (\#x_{1}, \#x_{2} )= (t_{1}, t_{2})) \tag{Since we assume $t_{1}+ t_{2} = s$}\\
	&= \sum_{d_{1}> d_{2} > \cdots > d_{s}} \Pr((\#x_{3},\cdots, \#x_{\bins}) = (t_{3}, \cdots, t_{\bins})\ |\ (\#x_{1}, \#x_{2} )= (t_{1}, t_{2}) \land (D_{1}, \cdots, D_{s})  = (d_{1}, \cdots, d_{s})) \\
	&\tab \tab \tab \times \Pr((D_{1}, \cdots, D_{s})  = (d_{1}, \cdots, d_{s})\ |\  (\#x_{1}, \#x_{2} )= (t_{1}, t_{2}))\tag{By total probability,}\\
	\end{align*}}
 Since each vote is independent,  $(\#x_{3},\cdots, \#x_{\bins})$ is independent of $ (\#x_{1}, \#x_{2})$. Moreover, the vote indices $(D_{1}, \cdots, D_{s})$ are independent of $(\#x_{1}, \#x_{2})$. As votes are i.i.d.,  $(\#x_{1}, \#x_{2})$ does not depend on the value of $(d_{1}, \cdots, d_{s})$. Thus,
	\begin{align*}
	&\Pr((\#x_{3},\cdots, \#x_{\bins}) = (t_{3}, \cdots, t_{\bins}) \ |\ \#x_{1} + \#x_{2} = s\land (\#x_{1}, \#x_{2})= (t_{1}, t_{2}))\\
	&= \sum_{d_{1}> d_{2} > \cdots > d_{s}} \Pr((\#x_{3},\cdots, \#x_{\bins}) = (t_{3}, \cdots, t_{\bins})\ | (D_{1}, \cdots, D_{s})  = (d_{1}, \cdots, d_{s}))\\
	&\tab \tab \tab \times \Pr((D_{1}, \cdots, D_{s})  = (d_{1}, \cdots, d_{s}))
	\end{align*}
	This concludes that the left hand side and right hand side probabilities of Equation~\ref{eqn:cond-indep} are equal. The random variables $(\#x_{1}, \#x_{2})$ are independent conditioned on $(\#x_{1}, \#x_{2})$.
\end{proof}

\section{Concrete Estimate of the Privacy Parameters}\label{section:empirical-results}
In this section we present an example of computing concrete estimates of $(0, \delta, \Delta)$-exact DDP values for several GSRs. For this example, we let $\Delta = \{\pi\}$ such that $\pi\in \Pi(\{x_1, x_2, x_3\})$ and $\pi(x_i) = \pi(x_j) = 1/3$ (i.e., votes are i.i.d. and uniform).  

We generated these concrete estimates by doing an exhaustive search of all possible profiles for 3 candidates and $n \leq 50$ votes, and computing the $\delta$ values exactly for each $n$. Since we know that $\delta = \Theta(1/\sqrt{n})$, we fit these values to $\delta(n) = \frac{1}{\sqrt{an + b}}$ via linear regression. 
We rank voting rules from most to least private, by the value $a$ for outputting the winner. The larger the $a$, the smaller the $\delta$ value and more private.
The resulting ranking from most to least private is: 

\hfil 2-approval $\rhd$ Plurality $\rhd$  Maximin  $\rhd$  STV  $\rhd$  Borda \hfill

We show in Table~\ref{GSR-concrete-delta-big} the fitted $\delta$ curves with the mean square error in the fit.

\begin{table}
	\centering
	\begin{tabular}{ |c|c|c|} 
		\hline
		Rule & Winner & Mean Square Error $(n \in [50])$  \\
		\hline
				Borda & $\delta(n) = \dfrac{1}{\sqrt{1.347n + 0.5263}}$  & 0.0566844201243\\ 
		\hline
				STV &  $\delta(n) = \dfrac{1}{\sqrt{1.495n + 0.02669}}$ & 0.0542992943035 \\
		\hline
				Maximin & $\delta(n) = \dfrac{1}{\sqrt{1.553n + 4.433}}$ & 0.0377631805983 \\
		\hline
			Plurality &  $\delta(n) = \dfrac{1}{\sqrt{1.717n - 0.09225}}$ & 0.0477175838906 \\ 
		\hline
			2-approval & $\delta(n) = \dfrac{1}{\sqrt{1.786n + 0.3536}}$   & 0.0454223047191 \\ 
		\hline
	\end{tabular}
	\vspace{2mm}
	\caption{$\delta$ values in $(0, \delta, \Delta)$-eDDP for some commonly-used voting rules under the i.i.d.~uniform distribution. $m=3$ and $n=10$ to $50$.}
	\label{GSR-concrete-delta-big}
\end{table}

%% file: main.bbl
\begin{thebibliography}{}

\bibitem[\protect\citeauthoryear{Bassily and Smith}{2015}]{Bassily2015}
Raef Bassily and Adam Smith.
\newblock {Local, Private, Efficient Protocols for Succinct Histograms}.
\newblock {\em STOC}, 2015.

\bibitem[\protect\citeauthoryear{Bassily \bgroup \em et al.\egroup
  }{2013}]{Bassily2013}
Raef Bassily, Adam Groce, Jonathan Katz, and Adam Smith.
\newblock {Coupled-worlds privacy: Exploiting adversarial uncertainty in
  statistical data privacy}.
\newblock {\em FOCS}, 2013.

\bibitem[\protect\citeauthoryear{Bhaskar \bgroup \em et al.\egroup
  }{2011}]{Bhaskar2011}
Raghav Bhaskar, Abhishek Bhowmick, Vipul Goyal, Srivatsan Laxman, and Abhradeep
  Thakurta.
\newblock {Noiseless Database Privacy}.
\newblock {\em Asiacrypt}, 7073, 2011.

\bibitem[\protect\citeauthoryear{Birrell and Pass}{2011}]{Birrell2011}
Eleanor Birrell and Rafael Pass.
\newblock {Approximately strategy-proof voting}.
\newblock {\em IJCAI}, 2011.

\bibitem[\protect\citeauthoryear{Blum \bgroup \em et al.\egroup
  }{2008}]{Blum2008}
Avrim Blum, Katrina Ligett, and Aaron Roth.
\newblock {A learning theory approach to noninteractive database privacy}.
\newblock {\em STOC}, 2008.

\bibitem[\protect\citeauthoryear{Bradshaw and
  Howard}{2018}]{bradshaw2018challenging}
Samantha Bradshaw and Philip~N Howard.
\newblock Challenging truth and trust: A global inventory of organized social
  media manipulation.
\newblock {\em The Computational Propaganda Project}, 2018.

\bibitem[\protect\citeauthoryear{Caragiannis \bgroup \em et al.\egroup
  }{2014}]{Caragiannis14:Modal}
Ioannis Caragiannis, Ariel~D. Procaccia, and Nisarg Shah.
\newblock {Modal Ranking: A Uniquely Robust Voting Rule}.
\newblock In {\em AAAI}, 2014.

\bibitem[\protect\citeauthoryear{Clayworth and Noble}{2016}]{iowa2016}
Jason Clayworth and Jason Noble.
\newblock Iowa caucus coin flip count unknown.
\newblock {\em OnPolitics}, 1:2016, 2016.

\bibitem[\protect\citeauthoryear{Duan}{2009}]{Duan2009}
Yitao Duan.
\newblock {Privacy without Noise}.
\newblock {\em CIKM}, 2009.

\bibitem[\protect\citeauthoryear{Dwork and Roth}{2014}]{Dwork2014}
Cynthia Dwork and Aaron Roth.
\newblock {The Algorithmic Foundations of Differential Privacy}, 2014.

\bibitem[\protect\citeauthoryear{Dwork \bgroup \em et al.\egroup
  }{2006}]{Dwork2006a}
Cynthia Dwork, F.~McSherry, Kobbi Nissim, and Adam Smith.
\newblock {Calibrating noise to sensitivity in private data analysis}.
\newblock {\em TCC}, 2006.

\bibitem[\protect\citeauthoryear{Dwork}{2006}]{Dwork2006}
Cynthia Dwork.
\newblock {Differential Privacy}.
\newblock {\em ICALP}, 2006.

\bibitem[\protect\citeauthoryear{Garfinkel \bgroup \em et al.\egroup
  }{2018}]{garfinkel2018understanding}
Simson Garfinkel, John~M Abowd, and Christian Martindale.
\newblock Understanding database reconstruction attacks on public data.
\newblock {\em Queue}, 2018.

\bibitem[\protect\citeauthoryear{Georges-Th{\'e}odule}{1952}]{Georges-Theodule1952:Les-theories}
Guilbaud Georges-Th{\'e}odule.
\newblock Les th{\'e}ories de l'int{\'e}r{\^e}t g{\'e}n{\'e}ral et le
  probl{\`e}me logique de l'agr{\'e}gation [theories of the general interest
  and the logical problem of aggregation].
\newblock {\em Economie appliqu{\'e}e}, 1952.

\bibitem[\protect\citeauthoryear{Ghosh \bgroup \em et al.\egroup
  }{2009}]{Ghosh2009}
Arpita Ghosh, Tim Roughgarden, and Mukund Sundararajan.
\newblock {Universally utility-maximizing privacy mechanisms}.
\newblock {\em STOC}, 2009.

\bibitem[\protect\citeauthoryear{Greene}{2003}]{greene2003econometric}
William~H Greene.
\newblock {\em Econometric analysis}.
\newblock Pearson Education India, 2003.

\bibitem[\protect\citeauthoryear{Groce}{2014}]{Groce2014}
Adam Groce.
\newblock {New Notions and Mechanisms for Statistical Privacy, PhD Thesis},
  2014.

\bibitem[\protect\citeauthoryear{Hall \bgroup \em et al.\egroup
  }{2012}]{Hall2012}
Rob Hall, Alessandro Rinaldo, and Larry Wasserman.
\newblock {Random Differential Privacy}.
\newblock {\em Journal of Privacy and Confidentiality}, 2012.

\bibitem[\protect\citeauthoryear{Hardt and Talwar}{2010}]{Hardt2010}
Moritz Hardt and Kunal Talwar.
\newblock {On the Geometry of Differential Privacy}.
\newblock {\em STOC}, 2010.

\bibitem[\protect\citeauthoryear{Hay \bgroup \em et al.\egroup
  }{2017}]{Hay2017}
M.~Hay, L.~Elagina, and G.~Miklau.
\newblock {Differentially private rank aggregation}.
\newblock {\em SDM}, 2017.

\bibitem[\protect\citeauthoryear{Horn and Johnson}{1990}]{horn1990matrix}
Roger~A Horn and Charles~R Johnson.
\newblock {\em Matrix analysis}.
\newblock Cambridge university press, 1990.

\bibitem[\protect\citeauthoryear{Kasiviswanathan and
  Smith}{2008}]{Kasiviswanathan2008}
Shiva~Prasad Kasiviswanathan and Adam Smith.
\newblock A note on differential privacy: Defining resistance to arbitrary side
  information.
\newblock {\em CoRR abs/0803.3946}, 2008.

\bibitem[\protect\citeauthoryear{Lee}{2015}]{Lee2015}
David~T. Lee.
\newblock {Efficient, private, and e-strategy proof elicitation of tournament
  voting rules}.
\newblock {\em IJCAI}, 2015.

\bibitem[\protect\citeauthoryear{Leung and Lui}{2012a}]{Leung2012a}
Samantha Leung and Edward Lui.
\newblock {Bayesian mechanism design with efficiency, privacy, and approximate
  truthfulness}.
\newblock {\em International Workshop on Internet and Network Economics},

\bibitem[\protect\citeauthoryear{Leung and Lui}{2012b}]{Leung2012}
Samantha Leung and Edward Lui.
\newblock {Bayesian mechanism design with efficiency, privacy, and approximate
  truthfulness}.
\newblock {\em WINE}, 2012.

\bibitem[\protect\citeauthoryear{Mattei and Walsh}{2013}]{Mattei13:Preflib}
Nicholas Mattei and Toby Walsh.
\newblock {PrefLib: A Library of Preference Data}.
\newblock In {\em {Algorithmic
  Decision Theory}}, {Lecture Notes in Artificial Intelligence}, 2013.

\bibitem[\protect\citeauthoryear{McSherry and Talwar}{2007}]{McSherry2007}
Frank McSherry and Kunal Talwar.
\newblock {Mechanism Design via Differential Privacy}.
\newblock {\em FOCS}, 2007.

\bibitem[\protect\citeauthoryear{Shang \bgroup \em et al.\egroup
  }{2014}]{Shang2014}
Shang Shang, Tiance Wang, Paul Cuff, and Sanjeev Kulkarni.
\newblock {The Application of Differential Privacy for Rank Aggregation:
  Privacy and Accuracy}.
\newblock {\em Information Fusion}, 2014.

\bibitem[\protect\citeauthoryear{{US Census
  Bureau}}{2019}]{uscensusbureau_2019}
{US Census Bureau}.
\newblock Voting and registration tables, Jun 2019.

\bibitem[\protect\citeauthoryear{Varah}{1975}]{varah1975lower}
James~M Varah.
\newblock A lower bound for the smallest singular value of a matrix.
\newblock {\em Linear Algebra and its Applications}, 1975.

\bibitem[\protect\citeauthoryear{Verini}{2007}]{verini_2007}
James Verini.
\newblock Big brother inc.
\newblock {\em Vanity Fair}, Dec 2007.

\bibitem[\protect\citeauthoryear{Wang \bgroup \em et al.\egroup
  }{2019}]{Wang2019:Practical}
Jun Wang, Sujoy Sikdar, Tyler Shepherd, Zhibing Zhao, Chunheng Jiang, and
  Lirong Xia.
\newblock {Practical Algorithms for STV and Ranked Pairs with Parallel
  Universes Tiebreaking}.
\newblock In {\em AAAI}, 2019.

\bibitem[\protect\citeauthoryear{Xia and Conitzer}{2008}]{Xia2008}
Lirong Xia and Vincent Conitzer.
\newblock Generalized scoring rules and the frequency of coalitional
  manipulability.
\newblock In {\em Electronic Commerce}, 2008.

\bibitem[\protect\citeauthoryear{Xia and Conitzer}{2009}]{Xia2009}
Lirong Xia and Vincent Conitzer.
\newblock {Finite Local Consistency Characterizes Generalized Scoring Rules}.
\newblock {\em IJCAI}, 2009.

\bibitem[\protect\citeauthoryear{Xia}{2015}]{Xia2015}
Lirong Xia.
\newblock {Generalized Decision Scoring Rules: Statistical, Computational, and
  Axiomatic Properties}.
\newblock {\em EC}, 2015.

\end{thebibliography}
